\documentclass[11pt,onecolumn]{IEEEtran}

\usepackage{times}
\usepackage{amsthm}
\usepackage{amsmath}    
\usepackage{amssymb}    
\usepackage{mathrsfs}   
\usepackage{epsfig}     
\usepackage{graphicx}
\usepackage{amstext}
\usepackage{epstopdf}
\usepackage{psfrag}

\IEEEoverridecommandlockouts

\newcommand{\cA}{{\cal A}}

\newcommand{\cC}{{\cal C}}

\newcommand{\cE}{{\cal E}}

\newcommand{\cI}{{\cal I}}

\newcommand{\cM}{{\cal M}}
\newcommand{\cN}{{\cal N}}

\newcommand{\cR}{{\cal R}}

\newcommand{\bS}{\mathbf{S}}

\newcommand{\bV}{\mathbf{V}}

\newcommand{\bX}{\mathbf{X}}
\newcommand{\bY}{\mathbf{Y}}

\newcommand{\bx}{\mathbf{x}}
\newcommand{\by}{\mathbf{y}}
\newcommand{\bs}{\mathbf{s}}
\newcommand{\bu}{\mathbf{u}}

\renewcommand{\le}{\leqslant}

\renewcommand{\ge}{\geqslant}

\newcommand{\dfn}{\stackrel{\triangle}{=}}

\newtheorem{thm}{Theorem} 
\newtheorem{cor}{Corollary}
\newtheorem{lem}{Lemma}
\newtheorem{prop}{Proposition}

\theoremstyle{definition}

\theoremstyle{definition}
\newtheorem{defn}{Definition}

\newcommand{\be}[1]{\begin{equation}\label{#1}}
\newcommand{\ee}{\end{equation}}
\newcommand{\eq}[1]{(\ref{#1})}

\newcommand{\Tref}[1]{Theo\-rem\,\ref{#1}}
\newcommand{\Pref}[1]{Pro\-po\-si\-tion\,\ref{#1}}
\newcommand{\Lref}[1]{Lem\-ma\,\ref{#1}}
\newcommand{\Cref}[1]{Co\-ro\-lla\-ry\,\ref{#1}}

\newcommand{\Fref}[1]{Figure~\ref{#1}}

\newcommand{\ep}{\epsilon}

\newcommand{\norm}[1]{\|#1\|}

\newcommand{\E}{\mathbb{E}}
\newcommand{\R}{\mathbb{R}}

\newcommand{\RDF}{R_{\text{DF}}}
\newcommand{\RAF}{R_{\text{AF}}}
\newcommand{\RBAF}{R_{\text{BAF}}}
\newcommand{\RBinary}{R_{\text{BSPDF}}}
\newcommand{\RTernary}{R_{\text{TSPDF}}}
\newcommand{\h}{{\rm h}}
\newcommand{\I}{{\rm I}}
\newcommand{\Cupper}{C_{\text{up}}}

\newcommand{\plotsize}{0.72}

\begin{document}


\title
{
Using Superposition Codebooks and Partial Decode and Forward in Low SNR Parallel Relay Networks 
}

\author{Farzad Parvaresh and Ra\'ul Etkin

\thanks{F. Parvaresh and R. Etkin are with Hewlett-Packard Laboratories, Palo Alto, CA 94304, USA. (emails: \{parvaresh, raul.etkin\}@hp.com)}
}

\maketitle


\begin{abstract}
A new communication scheme for Gaussian parallel relay networks based on superposition coding and partial decoding at the relays is presented. Some specific examples  are proposed in which two codebook layers are superimposed. The first level codebook is constructed with symbols from a binary or ternary alphabet while the second level codebook is composed of codewords chosen with Gaussian symbols. The new communication scheme is a generalization of decode-and-forward, amplify-and-forward, and bursty-amplify-and-forward. The  asymptotic low SNR regime is studied using achievable rates and minimum energy-per-bit as performance metrics. It is shown that the new scheme outperforms all previously known schemes for some channels and parameter ranges.
\end{abstract}


\section{Introduction}
Cooperation in wireless networks is often modeled through relay networks. In these models
one or more source nodes communicate with one or more destination nodes with the help of 
intermediate relay nodes. Relay network models are of relevance in wireless sensor networks where
the sensor nodes have limited transmission power capabilities and communication range. 
In this work we study a special type of relay network called the {\em diamond network} for the network with two relays and {\em parallel relay network} for the general case with more than two relays, introduced by Schein and Gallager \cite{SG00,S01}. This is a single-source single-destination layered network model with three layers, namely the source layer, relay layer, and destination layer, with the property that the nodes in each layer can only communicate with nodes in the next layer. The communication from the source to the relays takes place over a broadcast channel (BC) while the communication from the relays to the destination takes place over a multiple access channel (MAC). Note that in this simple model there is no direct communication path from the source to the destination. We focus on the Gaussian case, in which the transmitted signals are real and subject to an average power constraint, and the received signals are affected by channel attenuation and additive white Gaussian noise.

Despite the simplicity of the model, its capacity, that is, the maximum reliable communication rate from source to destination, is in general unknown. The best known capacity upper bound is based on the cut-set bound. Many communication strategies have been proposed, such as decode-and-forward (DF), compress-and-forward, amplify-and-forward (AF), bursty-amplify-and-forward (BAF), rematch-and-forward, and combinations thereof, leading to various capacity lower bounds. Depending on the channel gains and signal to noise ratios (SNR), some communication schemes perform better than others. Many wireless sensor networks operate in the low SNR regime where power is a scarce resource and the main performance limiter. As a result, it is of interest to design and evaluate communication schemes that exhibit good performance at low SNR. We study the performance of various communication schemes in the asymptotic regime of SNR going to zero using as performance metrics channel capacity and minimum energy-per-bit (i.e. minimum required energy to communicate 1 bit of information).

The first two communication schemes considered for parallel relay networks are DF and AF. In DF all the relays decode the source message and retransmit it to the destination achieving a beamforming gain. DF exhibits good performance when the links from the source to the relays are sufficiently stronger than the links from the relays to the destination. However, when the links from the source to the relays are relatively weak compared to the links from the relays to the destination, the requirement of decoding the source message at all relays is too restrictive. 
In AF the relays amplify the received signals with certain gains making no attempt to decode the source message. While amplifying the received signals the relays also amplify the received noise. While in some regimes AF can achieve the channel capacity, at low SNR the relays use most of their available power amplifying the received noise, and the resulting performance of AF is poor. 

An approach to improve the performance of AF in the low SNR regime is to remain silent a fraction of the symbol times and communicate in bursts to increase the effective SNR of the signal transmitted during the bursts. This communication scheme, denoted BAF and proposed in \cite{S01}, achieves better rates than DF when the links from the relays to the destination are sufficiently stronger than the links from the source to the relays. In BAF, the source, relays, and destination know in advance the burst pattern, which is usually a contiguous block of symbols. As a result, the burst pattern does not carry any information.

In this work we propose a new communication scheme tailored to the low SNR regime which we denote superposition-partially-decode-and-forward (SPDF). In the binary form of our proposed SPDF scheme, we generate a binary codebook (i.e. with symbols chosen from $\{0,1\}$) of sparse binary patterns, and encode part of the source message through the selection of a codeword from this codebook. The remaining part of the source message is encoded through an independent codeword, transmitted in the symbol positions where the binary codeword has ones. The relays then decode the binary codeword and amplify the received signal only in the time indices where the binary codeword has ones. Finally the destination jointly decodes both codewords. We see that SPDF leverages the benefits of BAF of increasing the symbol SNR to avoid excessive noise amplification, while achieving some additional rate by conveying information through the binary codeword. As will be shown in Sections~\ref{sec:achievable} and \ref{sec:analysis} the binary SPDF scheme can be generalized in various ways.

We characterize the rates that can be achieved with the SPDF scheme, and analyze its performance in the asymptotic regime of SNR$\to 0$ for parallel relay networks with some symmetry properties. We also study the  energy-per-bit achievable with SPDF and obtain a minimum energy-per-bit characterization within a constant multiplicative factor for symmetric Gaussian diamond networks. Our results show that the new communication scheme matches or exceeds the performance of all previously known communication schemes in the low SNR regime. 

The remainder of the paper is organized as follows. Subsections \ref{ssec:related} and \ref{ssec:notation} below, present related literature and the notation used in the rest the paper. In Section~\ref{sec:model} we describe the communication model that we will use throughout the paper. Section~\ref{sec:achievable} presents the achievable rates with known communication schemes and with the new proposed SPDF communication scheme. In Section~\ref{sec:upper_bound} cut-set upper bounds for various networks are presented, which will be used in later sections to evaluate the performance of the different achievable rates. In Section~\ref{sec:analysis} we analyze the performance of the existing and new communication schemes in the asymptotic low SNR regime through a capacity formulation. Section~\ref{sec:ebno} shows how the asymptotic low SNR capacity characterizations can be used to obtain minimum energy-per-bit performance bounds. Finally in Section~\ref{sec:conclusion} we give some concluding remarks. The proofs of the results are given in the appendices.

\subsection{Related work}
\label{ssec:related}
To put our work into context, we briefly discuss related results which provide various characterizations of the fundamental communication limits in relay networks.
While in general the upper and lower capacity bounds for relay networks do not match, the additive gap between them has been shown to be upper bounded by a constant that is independent of the channel gains and power constraints and that only depends on the number of nodes in the network \cite{ADT11,OD10}\footnote{This constant additive gap capacity characterization applies to general single source relay networks.}. For the symmetric $N$-relay Gaussian parallel network, the constant additive gap characterization has been improved in \cite{ND10} to 1.8 bits, independent of the number of relays $N$. Capacity characterizations within a constant additive gap are of practical relevance in moderate to high signal to noise ratio (SNR) regimes and in networks of small size, in which the constant additive gap is small compared to the achievable communication rates. 

Many wireless sensor networks operate in the low SNR regime, for which a capacity approximation within a constant multiplicative factor, instead of within a constant additive gap, is more appropriate. For general relay networks with single-source multicast Avestimehr et al. \cite{ADT11} derived a capacity approximation within a factor of $2 d(d+1)$ where $d$ is the maximum node degree in the network. In the specific case of a parallel relay network the maximum degree is $d=N$ so the above result becomes $2 N(N+1)$.  Recently, this capacity characterization within a multiplicative factor has been improved in \cite{ND10} for the special case of the symmetric $N$-relay Gaussian parallel network. In this special case, the capacity has been characterized  to within a multiplicative factor of 14, which is independent of the number of relay nodes or the maximum degree in the network. We note that this multiplicative factor can be readily improved for networks of small size where the exact cut-set bound can be computed.

The low SNR regime can also be studied via a minimum energy-per-bit formulation. The minimum energy-per-bit in the relay channel (i.e. single relay network with a direct signal path from source to destination) has been studied in \cite{EMZ06}. This work establishes a relationship between the minimum energy-per-bit and capacity, and characterizes the minimum energy-per-bit within a constant factor. 

The BAF communication scheme is not known to achieve capacity in relay channels with constant gains. However, it has been shown in \cite{AT07} that in the slow-fading Gaussian relay channel at low SNR, BAF achieves the outage capacity for low outage probability.

A communication scheme denoted rematch-and-forward (RM) has been proposed in \cite{KKEZ08} for cases where there is a bandwidth mismatch between the source-relay and relay-destination links. In addition, RM has been shown to offer performance benefits even in cases where there is no bandwidth mismatch as long as the SNR is sufficiently high. Furthermore, RM has been studied in the context of half-duplex diamond relay networks that allow communication between relays \cite{RGK09}.

\subsection{Notation}
\label{ssec:notation}
Regarding notation, we use lowercase letters to denote scalars, uppercase letters to denote random variables, boldface letters to denote vectors, and calligraphic uppercase letters to denote sets. For example, $a$ is a constant scalar, $\bV$ is a random vector, and $\mathcal{M}$ is a set. $\cA^c$ and $|\cA|$ denote the complement and the cardinality of the set $\cA$. An $n$-vector $(x_1,\ldots,x_n)$ is written as $\bx^n$, and its $t^{th}$ element is expressed as $x[t]$. Differential entropy and mutual information are denoted by $\h$ and $\I$ and are expressed in bits. $\ln$ is used for natural logarithm while $\log_2$ is used for base-2 logarithm. Probability and expectation are denoted by $\Pr$ and $\E$. A Gaussian distribution with mean $\mu$ and variance $\sigma^2$ is denoted by $\cN(\mu,\sigma^2)$. We use the notation $a_n \doteq 2^{n(b \pm \epsilon)}$ to express that $\left| \frac{1}{n} \log a_n - b \right| < \epsilon$.
$1\{\cdot\}$ is the indicator function.
All the rates are presented in bits. 


\section{Communication model}
\label{sec:model}

\begin{figure}[tb]

\centering
\includegraphics[height=0.30\textwidth]{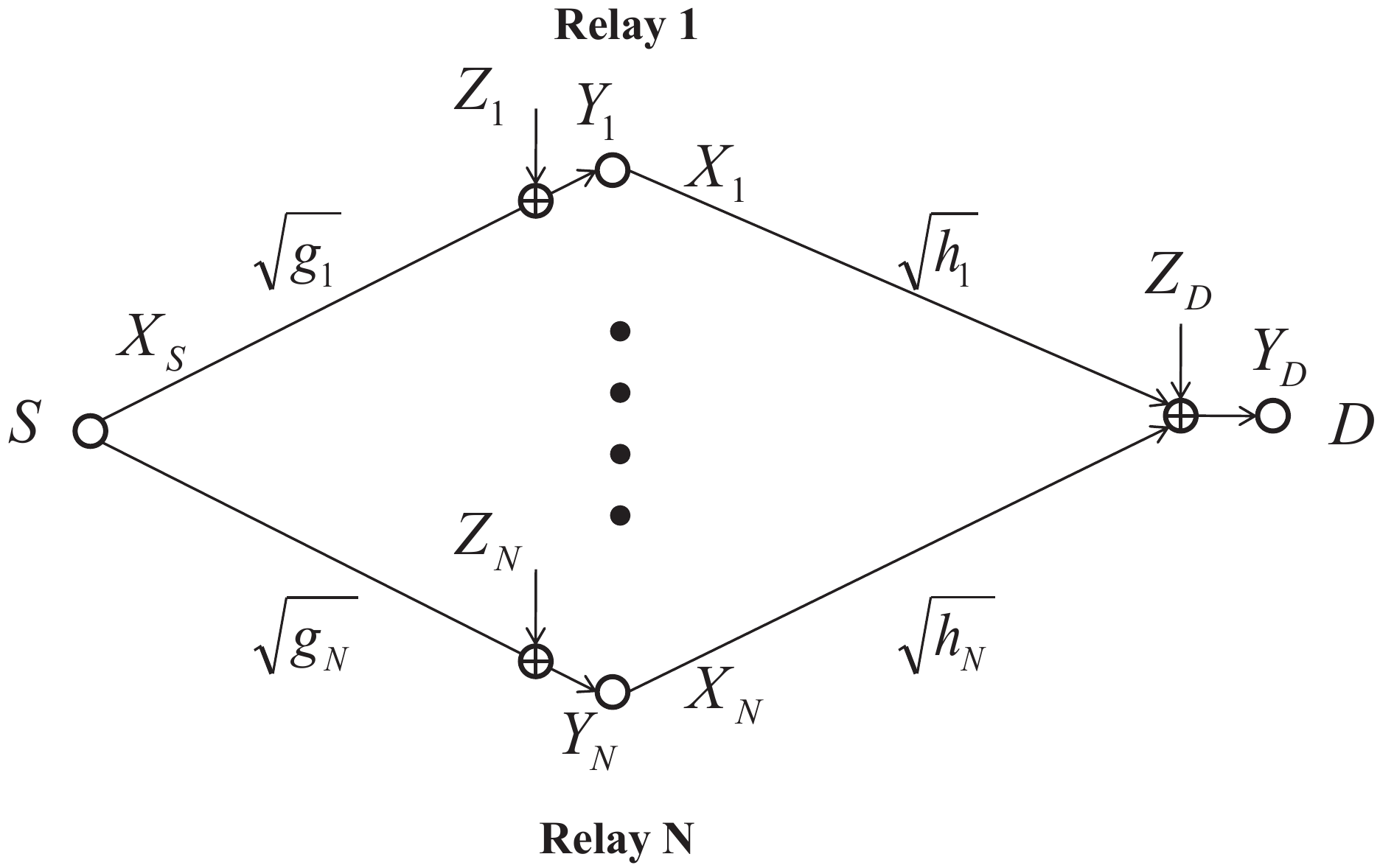}

\caption{Communication model for the Gaussian $N$-relay parallel network.}
\label{fig:channelN}

\end{figure}

We consider a parallel relay network consisting of a source node $S$, $N$ relay nodes Relay~$i$, $i=1,\ldots, N$,  and a destination node $D$ (cf. \Fref{fig:channelN}). We focus on the memoryless discrete-time Gaussian version of the model where the source and relays transmit real signals subject to given power constraints, and the received signals are attenuated by real channel gains and corrupted by independent additive white Gaussian noise.

We denote by $X_S[t]$ and $X_i[t]$ the signals transmitted by the source and relays $i=1,\ldots,N$ at discrete time $t$ respectively. 
We denote by $\sqrt{g_i} \in \R_+$ the channel gain from the source to relay $i$ and by $\sqrt{h_i}\in \R_+$ the channel gain from relay $i$ to the destination, noting that there is no loss of generality in considering non-negative channel gains. 
In addition, without loss of generality (WLOG), we assume $g_1 \le g_2\le \ldots \le g_N$. 
The received signal
at relay $i$, $Y_i$, and the received signal at the destination, $Y_D$,
are given by
\begin{align}
\label{eq:Y_i}
Y_i[t] & = \sqrt{g_i} X_S[t] + Z_i[t] \ \text{for} \ i=1,\ldots,N, \\ 
\label{eq:Y_D}
Y_D[t] & = \sum_{i=1}^N \sqrt{h_i} X_i[t] + Z_D[t],
\end{align}
where $Z_1[t], \ldots, Z_N[t]$ and $Z_D[t]$ are i.i.d. 
random variables with distribution $\cN(0,N_0)$.

The source wants to communicate to the destination a random message $W$ that is uniformly distributed in the set $\{1, 2, \ldots, M\}$. A $(2^{n R}, n)$ code for the parallel relay network consists of a set of integers
$\cM = \{1,2,\ldots, 2^{nR} \}$ with $\lfloor 2^{nR}\rfloor = M$,
called the message set and the following encoding and decoding functions.
The source uses an encoding function $\text{Enc}_S: \{1, 2, \ldots, M\} \to \R^n$ to map any message $w\in\cM$ into a vector $(x_{S,w}[1],\ldots,x_{S,w}[n])= \text{Enc}_S(w)$. At time $t$ relay $i$ maps the input signals $\{y_i[1],\ldots, y_i[t-1]\}$ into an output signal $x_i[t]$ using an encoding function $\text{Enc}_{i}^{t-1} : \R^{t-1} \to \R$ for $i=1,\ldots,N$ and $t=1,\ldots, n$. The decoder outputs an estimate of the transmitted message $\hat w = \text{Dec}(y_D[1],\ldots,y_D[n])$ using the decoding function $\text{Dec} : \R^n \to \{1, 2, \ldots, M\}$.

We assume the average power of the transmitted signals to be upper bounded by $P_S$ at source node and $P_i$ at relay nodes $i=1,2,\ldots,N$, i.e. for a block of $n$ symbols
\begin{align}
\frac{1}{n} \sum_{t=1}^n x_{S,w}[t]^2  & < P_S,  \ \text{ for all } w \in \cM, \label{eq:pc1}\\
\frac{1}{n} \sum_{t=1}^n x_i[t]^2 & < P_i \text{ for all } \by_i^n\in \R^n, \text{ for } i=1,\ldots,N, 
\label{eq:pc2}
\end{align}
noting that in \eq{eq:pc2} $x_i[t]=\text{Enc}_i^{t-1}(\by_i^{t-1})$. 

The average probability of error is defined as the probability that
the decoder's estimate $\hat{W}$ is not equal to the transmitted message $W$, and is given by
$$
P_e^{(n)} = \frac{1}{2^{nR}} \sum_{w \in \cM}
\Pr \Big\{ \text{Dec}( \{ \bY_D[i] \}_{i=1}^n ) \neq w \ | \ w \text{ sent} \Big\},
$$
where $W$ is assumed to be uniformly distributed over the elements of $\cM$.

The rate $R$ is said to be achievable for the parallel
relay network if there exists a sequence of $(2^{nR},n)$ codes with the transmitted signals satisfying the power constraints (\ref{eq:pc1}) and (\ref{eq:pc2}) such that $P_{e}^{(n)}$ tends to zero as $n$ tends to infinity. 

For a code with block length $n$ and rate $R_n \ge 1/n$, where the rate can vary with $n$, 
the energy of codeword $w \in \cM$ is given by
$$
\cE^{(n)}_S(w) = \sum_{t=1}^n x_{S,w}[t]^2
$$
and the maximum transmission energy of relay $i$ is
$$
\cE^{(n)}_i = \sup_{\by_i^n\in \R^n} \left( \sum_{t=1}^n x_i[t]^2  \right).
$$
The energy-per-bit for the code is given by
$$
\cE^{(n)} = \frac{1}{n R_n}\left( \max_{w \in \cM} \cE^{(n)}_S(w) + 
\sum_{i=1}^N \cE_i^{(n)} \right).
$$
An energy-per-bit $\cE$ is said to be achievable if there exist a sequence of $(2^{nR_n},n)$ codes with $P_e^{(n)} \to 0$ and $\lim \sup \cE^{(n)} \le \cE$.
The minimum energy-per-bit $\cE_b$ is the infimum of the set of achievable energy-per-bit values.

In Sections \ref{sec:achievable} through \ref{sec:analysis} we assume that the power constraints on the source and relays are equal to one, i.e. $P_S = P_i = 1$, for $i=1,\ldots,N$ and the noise variances at the relays and destination are equal to one, i.e. $N_0=1$. This normalization does not reduce the generality of the results since different power constraints and noise variances can be absorbed in the channel gains.


\section{Achievable rates}
\label{sec:achievable}
\subsection{Existing communication schemes}

In this subsection we obtain the rates achievable with DF, AF, and BAF.
In DF each relay must decode the message transmitted by the source node. As a result, the achievable rate cannot exceed the capacity of the point-to-point channel from the source to each relay $i$. Therefore,
\begin{equation}
\label{eq:rdf1}
\RDF \le \frac{1}{2} \log_2(1+g_i), \quad i=1,\ldots,N.
\end{equation}
In addition, since all the relays decode the transmitted message, they can beamform their transmissions to the destination. The achievable rate cannot exceed the capacity of the multiple-input single-output (MISO) point-to-point channel from the relays to the destination, obtaining
\begin{equation}
\label{eq:rdf2}
\RDF \le \frac{1}{2} \log_2 \left[1+\left(\sum_{i=1}^N \sqrt{h_i} \right)^2\right].
\end{equation}
Combining (\ref{eq:rdf1}) and (\ref{eq:rdf2}) we obtain the achievable rate with the DF communication scheme
\begin{equation}
\label{eq:rdf3}
\RDF = \frac{1}{2} \log_2\left[1+\min\left\{\min_{1\le i\le N} g_i, \left(\sum_{i=1}^N \sqrt{h_i} \right)^2 \right\}\right].
\end{equation}

In AF, each relay amplifies its received signal by an amplification factor $\sqrt{\kappa_i}$. In order not to exceed the transmit power constraint, the amplification factor must satisfy $\kappa_i < 1/(1+g_i)$. The optimal values of $\{\kappa_i\}_{i=1}^N$ are obtained through optimization. AF also achieves a beamforming gain from the relays to the destination, but in contrast to DF, the received noise at the relays is amplified and added to the overall received noise at the destination. Taking into account the amplification factors $\sqrt{\kappa_i}$ and the additional noise appearing at the destination, we obtain the following achievable rate for the AF scheme

\begin{equation}
\label{eq:raf}
\RAF = 
\sup_{\substack{0 \le \kappa_i < 1/(1+g_i): \\ {i=1,\ldots,N}}} \frac{1}{2} \log_2 \left[1+\frac{\left(\sum_{i=1}^N \sqrt{\kappa_i g_i h_i}\right)^2}{\sum_{i=1}^N \kappa_i h_i +1} \right].
\end{equation}

Finally, in BAF the source and relays transmit only during a fraction $\delta$ of the time, remaining silent during the remaining channel uses. This allows to increase the transmission power at the source and relays by a factor $1/\delta$ while still satisfying the average power constraint. The amplification factor at relay $i$, $\kappa_i$, must satisfy $\kappa_i < 1/(\delta+g_i)$, with the best performance obtained by optimizing the rate over the feasible values of $\{\kappa_i\}_{i=1}^N$. Since the channel is only used a fraction $\delta$ of time, which can also be optimized, the resulting achievable rate of the BAF scheme is
\begin{equation}
\label{eq:rbaf}
\RBAF = \sup_{ \substack{0\le \delta \le 1 \\ {0\le \kappa_i < 1/(\delta+g_i):i=1,\ldots,N}}} \frac{\delta}{2} \log_2 \left[1+\frac{\left(\sum_{i=1}^N \sqrt{\frac{\kappa_i g_i h_i}{\delta}}\right)^2}{\sum_{i=1}^N \kappa_i h_i +1} \right].
\end{equation}


\subsection{The Superposition-Partially-Decode-and-Forward communication scheme}

Sending correlated signals at the relays achieves a beamforming gain in the links to the destination. This correlation can be achieved if the relays decode the source message and re-encode it using the same codebook (decode-and-forward), or if they amplify the received signals (amplify-and-forward). As mentioned in the introduction, AF leads to poor performance at low SNR due to noise amplification. The use of the AF scheme in bursts (bursty-amplify-and-forward) allows to increase the effective symbol SNR while satisfying the power constraints at the source and relays.

The SPDF communication scheme that we propose is a generalization of DF, AF, and BAF. 
In the BAF scheme, the bursts occur in predetermined time intervals, and no information is carried in the timing of the bursts. A burst pattern can be interpreted as a binary $\{0,1\}$ sequence that has ones in the time positions where the source uses positive transmission power. In BAF this binary sequence is constant. Our SPDF scheme is inspired by the observation that it is possible to encode information by using different binary patterns as long as the relays and destination can identify these patterns. The relays need to determine the binary pattern in order to know when to amplify the received signal and when to remain silent.

As in BAF, we increase the effective SNR of AF by transmitting sparse or low duty cycle signals. One way of generating these signals is by generating a binary codebook $\cC_1$ with random i.i.d. codewords $\bu_{1,i_1}$ with $\text{Bernoulli}(\delta)$ symbols, and for each of these codewords $\bu_{1,i_1}$ generating a random codebook $C_{2,i_1}$ with codewords $\bu_{2,i_1,i_2}$ with symbols distributed as $\cN(0,1/\delta)$ for the time indices where $u_{1,i_1}[t]=1$ and with constant value 0 for the time indices where $u_{1,i_1}[t]=0$. The source transmits a message indexed by $(i_1, i_2)$ by transmitting $\bu_{2,i_1,i_2}$. Then the relays attempt to decode the binary codeword $\bu_{1,i_1}$ and amplify the received signal only in the time indices $t$ where $u_{1,i_1}[t]=1$, using appropriate gains to satisfy the power constraints. Finally the receiver attempts to decode both $\bu_{1,i_1}$ and $\bu_{2,i_1,i_2}$ generating an estimate $(\hat{i}_1, \hat{i}_2)$. We call this communication scheme Binary SPDF or BSPDF. 

The following theorem characterizes the rates that can be achieved with the BSPDF communication scheme. 
\begin{thm}
\label{thm:BSPDF}
Assume $B \sim \text{Bernoulli}(\delta)$ and $X_S| \{B = 0\} = 0$ (i.e. a mass of probability one at zero), 
$X_S | \{B = 1\} \sim \cN(0,\sigma^2)$. The received signal at relay $i=1,\ldots,N$,
is $Y_i = \sqrt{g_i} X_S + Z_i$, where $Z_i$ is i.i.d. $\cN(0,1)$. Relay $i$, $i=1,\ldots,N$,
amplifies the received signal $Y_i$ by a constant $\sqrt{\kappa_i}$ if $B$ is equal
to one, otherwise the relay sets its output equal to zero. This results in $Y_D=1\{B=1\}\sum_{i=1}^N \sqrt{\kappa_i h_i} Y_i+Z_D$. If $\delta \sigma^2 < 1$ and $\delta \kappa_i (g_i \sigma^2 + 1) < 1$, $i=1,\ldots,N$, then the rate $R=R_1+R_2$ is achievable if
\begin{align}
\label{eq:BCU1Vr}
R_1 & < \I(B; Y_1), \\
R_2 & < \I(X_S; Y_D|B), \\
\label{eq:BCU2VDU1}
R_1+R_2 & < \I(B, X_S; Y_{D}). 
\end{align}
\end{thm}
\begin{proof}
The proof is given in Appendix~\ref{app:proofBSPDF}.
\end{proof}

We note that in \Tref{thm:BSPDF} $\kappa_1,\ldots, \kappa_N, \sigma^2$ and $\delta$ are parameters that can be optimized for given channel gains in order to maximize the achievable rate $R$.

As will be shown in Section~\ref{sec:analysis}, timesharing between two
BSPDF schemes with different parameters can improve the achievable rate. This means that it is possible to improve performance by using different average power levels in the non-zero parts of the transmitted codewords, and using different amplification factors at the relays for these different codeword parts. This suggests that the BSPDF scheme outlined above can be generalized by considering more general distributions to generate $\cC_1$. In addition, we can generalize the mapping that the relays use to map at time $t$ the decoded symbol and received signal $(u_{1,i_1}[t], y_i[t])$ into $x_i[t]$, the transmitted signal.

The following theorem presents a generalization of \Tref{thm:BSPDF} to ternary alphabets for the symbols of the codewords in $\cC_1$ and two different amplification factors for the relaying functions of each relay. We call the resulting communication scheme Ternary SPDF or TSPDF.

\begin{thm}
\label{thm:TSPDF}
Consider the discrete random variable $T$ over the alphabet $\{0,1,2\}$ such that $\Pr(T = 0) = 1-\delta_1-\delta_2$, $\Pr(T = 1) = \delta_1$ and $\Pr(T = 2) = \delta_2$ for $\delta_1, \delta_2 > 0$, $\delta_1 + \delta_2 \le 1$. Let the random variable $X_S$ be such that
$X_S | \{ T = 0 \} = 0$ (i.e. a mass of probability one at zero), $X_S | \{ T = 1\} \sim \cN(0, \sigma_1^2)$ and
$X_S | \{T = 2\} \sim \cN(0,\sigma_2^2)$.
The received signal at relay $i=1,\ldots,N$, is given by $\sqrt{g_i} X_S + Z_i$ where $Z_i\sim \cN(0,1)$ are
i.i.d.. Relay $i$ amplifies the received signal
by a constant $\sqrt{k_{i1}}$ if $T = 1$, by a constant $\sqrt{\kappa_{i2}}$
if $T = 2$, and sets its output to zero if $T = 0$, for $i = 1,2,\ldots, N$. This results in $Y_D=1\{T=1\}\sum_{i=1}^N \sqrt{\kappa_{i1} h_i} Y_i+1\{T=2\}\sum_{i=1}^N \sqrt{\kappa_{i2} h_i} Y_i+Z_D$.
Given $\delta_1 \sigma_1^2 + \delta_2  \sigma_2^2 <1 $ and 
$\delta_1 \kappa_{i1}(g_i \sigma_1^2 + 1) + \delta_2 \kappa_{i2}(g_i \sigma_2^2+1)<1$ for $i = 1,2,\ldots, N$, then the rate $R = R_1 + R_2$ is achievable if
\begin{align}
R_1 & < \I(T; Y_1), \\
R_2 & < \I(X_S; Y_D|T), \\
R_1+R_2 & < \I(T, X_S; Y_D).
\end{align}
\end{thm}
\begin{proof}
The proof is given in Appendix~\ref{app:proofTSPDF}.
\end{proof}

We can generalize the BSPDF and TSPDF communication schemes by using more than two message levels and require different relays to decode different message-level sets. For example in cases where the channel gains $\sqrt{g_i}$ are different we may require some relays to decode all transmitted message levels, while requiring other relays to decode the message at level one, and finally let other relays amplify the received signals without decoding any messages. 

In general, we can consider $K$ message levels with $K$ up to $2^N$ (one message level for each subset of relays). However, the fact that the broadcast channel from the source to the relays is degraded implies that $N+1$ levels suffice. This is because if the message at level $k$ is to be decoded by relay $i$, all relays $j$ with $j > i$ are also able to decode the message due to the assumption $g_j \ge g_i$ for $j > i$. We use $f(i)$, $f:\{1,\ldots, N\} \to \{0,1,\ldots, K\}$, to denote the maximum message level that relay $i$ can decode or zero if relay $i$ does not decode any message level, noting that $f$ is non-decreasing. 
In this multi-level message setting, we require the destination be able to decode all message levels.

We can further generalize the BSPDF and TSPDF schemes by considering more general relaying functions. We denote by $w_i(u_1,\ldots,u_{f(i)},y_i)$ the transfer function of relay $i$ which at time $t$ produces the output $x_i[t]=w_i(u_1[t],\ldots,u_{f(i)}[t],y_i[t])$ based on the decoded codeword symbols $u_1,\ldots,u_{f(i)}$ and the received signal $y_i$ at time $t$. We can think of the symbols $u_1[t],\ldots,u_{f(i)}[t]$ as a directive to relay $i$ about how to process its input $y_i[t]$ in order to produce its output $x_i[t]$. The multiple codebook levels allow to give different directives to the various relays depending on their decoding capabilities. In addition, the codewords obtained from these symbols convey part of the source message to the destination.

\Tref{thm:SPDF} below establishes the rates that can be achieved with this more general communication scheme which we call SPDF. While these rates are characterized in terms of single letter bounds, the actual computation of the achievable rates requires appropriate choices for the joint distribution of the random variables used to generate the codebooks, the function $f(\cdot)$, and the mapping functions used at the relays, which are all design parameters.

\begin{thm}
\label{thm:SPDF}
Let $1\le K \le N+1$ be some integer constant and $f:\{1,\ldots, N\} \to \{0,1,\ldots, K\}$ be a non-decreasing function that denotes the maximum message level that relay $i$ can decode.
In addition, let $(U_1,\ldots,U_K)$ be distributed according to some joint distribution $P_{U_1,\ldots,U_K}$ satisfying $\E[U_K^2]< 1$, and define 
\begin{align}
V_i   =& \sqrt{g_i} U_K + Z_i, \text{ for } i=1,\ldots,N,\\
V_{D} =& \sum_{i=1}^N \sqrt{h_i} w_{i}(U_1,\ldots,U_{f(i)},V_i)+Z_{D},
\end{align}
where $Z_1,\ldots, Z_{N}$ and $Z_D$ are i.i.d. $\cN(0,1)$, and the relay transfer functions $w_i:\R^{f(i)+1} \to \R$ for $i=1,\ldots,N$ satisfy
\[
\E[(w_{i}(U_1,\ldots,U_{f(i)},V_i))^2] < 1, \text{ for } i=1,\ldots, N.
\]
Then, the rate $R = \sum_{m=1}^K R_m$ is achievable in the parallel relay network if
\begin{align}
\label{eq:thmSPDF}
\sum_{m=k}^{f(i)} R_m <& \I(U_{k},\ldots, U_{f(i)}; V_i | U_1,\ldots, U_{k-1}), \text{ where } i=\min_{1 \le j \le N:f(j)=k} j, \text{ for } k=1,\ldots, K, \nonumber\\
\sum_{m=k}^{K} R_m <& \I(U_{k},\ldots, U_{K}; V_D | U_1,\ldots, U_{k-1}), \text{ for } k=1,\ldots, K.
\end{align}
\end{thm}

\begin{proof}
The proof is given in Appendix~\ref{app:SPDF}.
\end{proof}

In order explain the role of the different variables and functions introduced in \Tref{thm:SPDF} we derive their relationship to the quantities introduced in \Tref{thm:BSPDF} for BSPDF. BSPDF uses $K=2$ codebook levels, and the following function $f$: $f(1)=\ldots=f(N)=1$ (i.e. all relay nodes decode the first level codeword, while the destination decodes the first and second level codewords). In addition, the binary auxiliary random variable $B$ of \Tref{thm:BSPDF} becomes $U_1$ in \Tref{thm:SPDF}, while $X_S$ becomes $U_2$. In \Tref{thm:BSPDF} we set the relay transfer functions $w_i(B,Y_i)=1\{B=1\} \sqrt{\kappa_i} Y_i$, for $i=1,\ldots, N$. We note that due to technical reasons in the proof of \Tref{thm:SPDF} the auxiliary random variables $(U_1,\ldots, U_K)$ have a continuous distribution. The proof of \Tref{thm:BSPDF} shows how to emulate the discrete auxiliary random variable of BSPDF with a continuous random variable.

We can also use \Tref{thm:SPDF} to generalize the BSPDF scheme by allowing different relays to decode different message levels. The function $f$, as before, specifies the maximum message level that relay $i$ can decode or zero if it does not decode any message level. We use the notation BSPDF($f$) to make the dependence on $f$ explicit.
\begin{cor}[BSPDF(f)]
\label{cor:BSPDF2}
Assume $B \sim \text{Bernoulli}(\delta)$ and $X_S| \{B = 0\} = 0$ (i.e. a mass of probability one at zero), 
$X_S | \{B = 1\} \sim \cN(0,\sigma^2)$. The received signal at relay $i=1,\ldots,N$,
is $Y_i = \sqrt{g_i} X_S + Z_i$, where $Z_i$ is i.i.d. $\cN(0,1)$. Relays $i$ with $f(i)=0$ amplify the received signals $Y_i$ by constants $\sqrt{\kappa_i}$, relays $i$ with $f(i)=1$ amplify the received signals $Y_i$ by constants $\sqrt{\kappa_i}$ if $B$ is equal to one, otherwise the relays set their output equal to zero, and relays $i$ with $f(i)=2$ retransmit the source message $X_S$. This results in $Y_D=\sum_{i:f(i)=0} \sqrt{\kappa_i h_i} Y_i + \sum_{i:f(i)=1}1\{B=1\} \sqrt{\kappa_i h_i} Y_i + \sum_{i:f(i)=2} \sqrt{h_i} X_S+Z_D$.
If $\delta \sigma^2 < 1$, $\kappa_i(\delta g_i \sigma^2  + 1) < 1$, for all $i$ such that $f(i)=0$, and $\delta \kappa_i (g_i \sigma^2 + 1) < 1$, for all $i$ such that $f(i)=1$, then the rate $R=R_1+R_2$ is achievable if
\begin{align}
R_1 & < \I(B; Y_{i}),   \ \text{ where } i=\min_{1\le j \le N:f(j)=1} j, \\
R_2 & < \I(X_S; Y_{i}|B),   \ \text { where }  i=\min_{1\le j \le N:f(j)=2} j, \\
R_2 & < \I(X_S; Y_D|B),   \\
R_1+R_2 & < \I(B, X_S; Y_{D})  . 
\end{align}
\end{cor}

\begin{proof}
The proof follows along the same lines of the proof of \Tref{thm:BSPDF} by specifying proper auxiliary random variables and relay transfer functions in \Tref{thm:SPDF}. The details are omitted.
\end{proof}

It will be shown in Section~\ref{sec:asym} that this BSPDF generalization provides rate improvements in asymmetric channels.

\subsection{Timesharing different schemes}
\label{ssec:timesharing}
It is possible to timeshare different communication schemes, splitting the available time and power among them. As will be seen in Section~\ref{sec:analysis}, for some channels timesharing various schemes achieves larger rates than those achievable with any of the constituent schemes alone. In this subsection we characterize the rates achievable with timesharing.

Let $R_i(P_{S}, P_{1},\ldots, P_{N})$, denote the maximum achievable rate with communication scheme $i$ under average power constraints at the source and relays $P_{S},P_{1},\ldots, P_{N}$ respectively. Define the achievable rate-power region with communication scheme $i$ by $\tilde \cR_i$, i.e.
\[
\tilde \cR_i = \{(r_i, p_S, p_1,\ldots, p_N)\in \R^{N+2}_+):r_i\le R_i(p_S, p_1,\ldots, p_N) \}.
\]
Timesharing various schemes indexed by $i\in \cI$ we can achieve rate-power vectors in
$\tilde \cR = \text{Convex Hull}(\cup_{i\in \cI} \tilde \cR_i)$. In the channel with unit power constraints at the source and relays we can achieve rate-power vectors in $\cR = \{(r, p_S, p_1,\ldots, p_N)\in \tilde \cR: p_S < 1, p_1 < 1, \ldots, p_N < 1 \}$.

For all the schemes that we consider in this paper $R_i(P_{S}, P_{1},\ldots, P_{N})$ is a continuous function. Since in addition $R_i(0, \ldots, 0)=0$ for any communication scheme, we have that $\cup_{i\in \cI} \tilde \cR_i$ is a connected set. Using Carath\'eodory's theorem \cite{Egg69} and its extension we have that any point $(r,p_s,p_1,\ldots,p_N)\in \tilde \cR$ can be written as $\sum_{j=1}^{N+2} \lambda_j (r_j,p_{s,j},p_{1,j},\ldots,p_{N,j})$ with $\lambda_j \ge 0$, $\sum_{j=1}^{N+2} \lambda_j=1$ and $(r_j,p_{s,j},p_{1,j},\ldots,p_{N,j}) \in \tilde \cR_i$ for some $i$, and for $j=1,\ldots, N+2$. This means that it is enough to timeshare $N+2$ pure schemes to achieve the timesharing region.

Furthermore, in some special cases we can reduce the number of pure schemes to timeshare below $N+2$. For example, in a symmetric parallel relay network we may restrict attention to communication schemes that assign the same power to all relays. In this case, we can reduce the dimension of the space where all the sets above (i.e. $\tilde\cR_i, \tilde \cR, \cR $) are defined to three, and timesharing three pure schemes suffices to achieve the timesharing region. 

The timesharing results of Section~\ref{sec:analysis} are obtained by using the above properties. 


\section{Cut-set upper bound}
\label{sec:upper_bound}
The best known upper bound for the parallel relay network is based on the so-called cut-set upper bound. For a general network with node set $\cN$ a cut $\Omega$ is any subset of $\cN$. Let $Y_i$ and $X_i$ be the input and output variables associated with node $i$, and $R_{ij}$ be the transmission rate from node $i$ to node $j$. The following proposition gives the cut-set capacity upper bound for the network \cite[Theorem 15.10.1]{Cover}.

\begin{prop}
\label{pro:cut_set}
If the information rates $\{R_{ij}\}$ are achievable, there exists some joint probability distribution $p(X_1,\ldots, X_{|\cN|})$ (or probability density function in case of continuous random variables) such that
$$
\sum_{i \in \Omega, j \in \Omega^c} R_{ij} \le \I(\bX_\Omega; \bY_{\Omega^c}|\bX_{\Omega^c})
$$
for all $\Omega \subseteq \cN$.
\end{prop}

\begin{figure}[tb]

\centering
\includegraphics[height=0.3\textwidth]{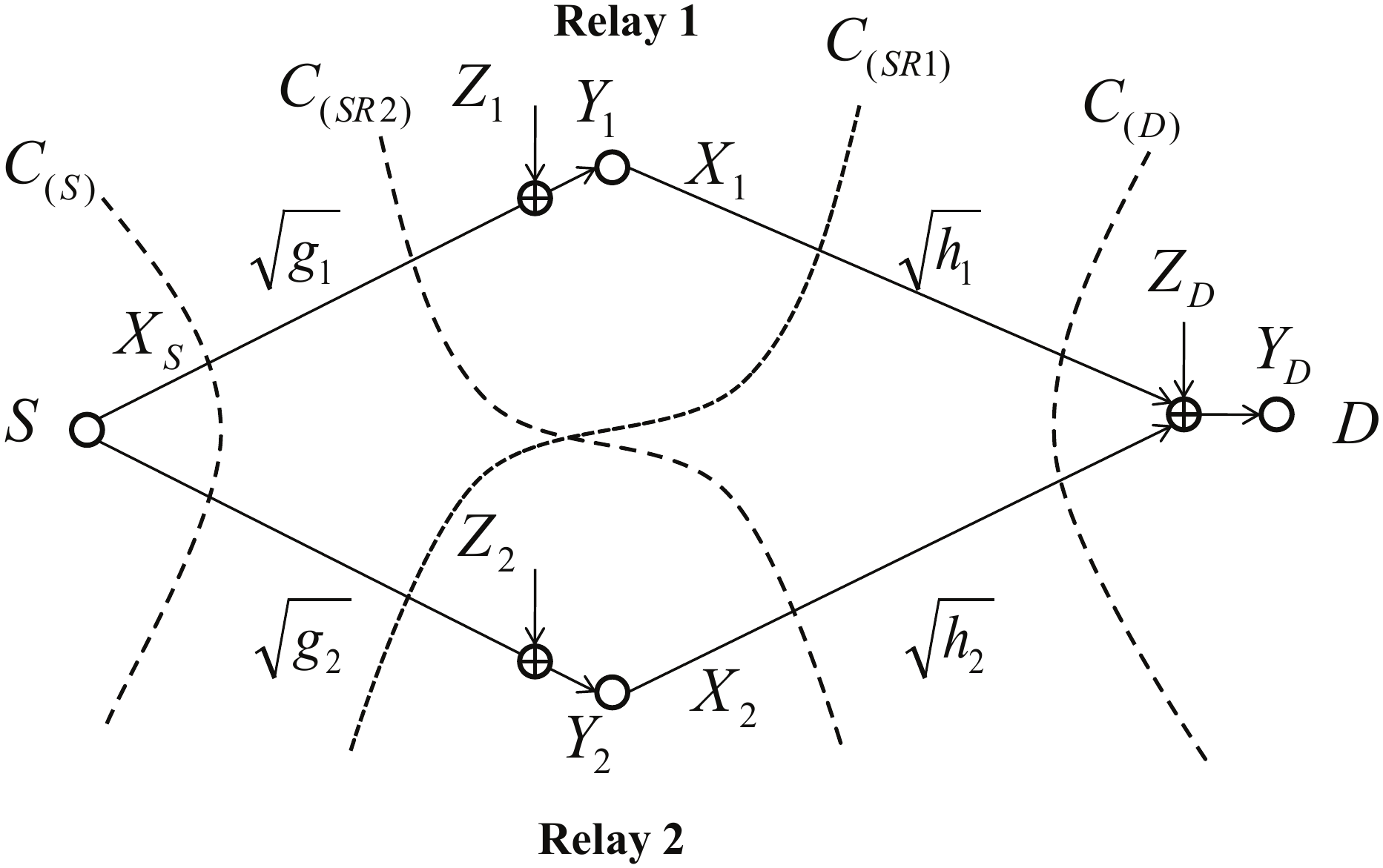}

\caption{Various cuts in a 2-relay diamond network.}
\label{fig:cut_set}

\end{figure}

For a 2-relay diamond network there are four cuts to consider, as shown in \Fref{fig:cut_set}. Specializing \Pref{pro:cut_set} to this network we obtain the following theorem (stated without proof).

\begin{prop}
\label{prop:cut_set}
The capacity $C$ of the Gaussian 2-relay diamond network is upper bounded by
\[
C\le \max_{0\le \rho \le 1} \min\left\{C_{(S)}, C_{(SR_1)}, C_{(SR_2)}, C_{(D)}  \right\},
\]
where
\begin{align*}
C_{(S)} = & \frac{1}{2} \log_2 \left(1+g_1+g_2\right),\\
C_{(SR_1)} = & \frac{1}{2} \log_2 \left[(1+g_2)\big(1+h_1\big(1-\rho^2\big)\big)\right],\\
C_{(SR_2)} = & \frac{1}{2} \log_2 \left[(1+g_1)\big(1+h_2\big(1-\rho^2\big)\big)\right],\\
C_{(D)} = & \frac{1}{2} \log_2 \left(1+h_1+h_2+2\rho\sqrt{h_1 h_2}\right).
\end{align*}
\end{prop}

While in principle \Pref{prop:cut_set} can be extended to networks with any number of relays, $N$, the computation of the cut-set bound becomes increasingly complex as $N$ increases. For symmetric networks where $g_1=g_2=\ldots=g_N=g$ and $h_1=h_2=\ldots=h_N=h$, an alternative upper bound is provided in \cite{ND10} as stated below.

\begin{prop}(\cite[Lemma 6]{ND10})
\label{pro:cut_setN}
For any $N \ge 2$, $g,h > 0$, the capacity of the Gaussian $N$-relay parallel network is upper bounded by
\[
C(N,g,h)\le \sup_{\rho \in [0,1)} \min_{n\in\{0,\ldots,N\}}\left(\frac{1}{2} \log_2\left(1+(N{-}\,n)g \right) +\frac{1}{2} \log_2\left(1+n\left(1{+}(n{-}\,1)\rho-\frac{n(N{-}\,n)\rho^2}{1{+}(N{-}\,n{-}\,1)\rho} \right) h \right)\right).
\]
\end{prop}


\section{Analysis of SPDF in the asymptotic low SNR regime}
\label{sec:analysis}
Theorems~\ref{thm:BSPDF} and \ref{thm:TSPDF} provide explicit single-letter expressions of achievable rates for BSPDF and TSPDF. However, finding the maximum communication rates achievable with these schemes  requires optimizing some parameters and computing differential entropies of mixtures of Gaussian random variables. 

To simplify the numerical computations and to gain more insight into the low SNR regime we consider an asymptotic regime where the channel gains tend to zero while maintaining a constant ratio among them. Since in this regime the channel capacity goes to zero, we measure performance by the ratio of the achievable rates and some reference channel gain. We consider two types of parallel relay networks:
\begin{itemize}
\item N-relay symmetric: $N\ge 2$, $g_1=g_2=\ldots=g_N=g$, $h_1=h_2=\ldots=h_N=h$. 
We let $h,g \to 0$ while keeping the ratio $g/h$ constant. We compare the performance of the different schemes by computing $R/g$.

\item 2-relay asymmetric: $N=2$, $g_1=h_2=g$, $h_1=g_2=h$. 
As before, we let $h,g \to 0$ while keeping the ratio $g/h$ constant. 
We compare the performance of the different schemes by computing $R/\sqrt{gh}$.
\end{itemize}

We note that we could have used as performance metric $R/x$ for any $x$ that goes to zero at the same rate as $g$ and $h$. The choice of $R/\sqrt{gh}$ as performance metric in the 2-relay asymmetric network results in equal performance for $h/g=\alpha$ and $h/g=1/\alpha$. This symmetry allows to present the results in a more compact form.

The benefit of considering the asymptotic low SNR regime is that we can approximate the differential entropy of a mixture of Gaussian distributions with an explicit expression using Taylor series, as shown in the following lemma.

\begin{lem}
\label{lem:b_h}
Let the probability distribution of the random variable $Y$ be a mixture
of $q$ Gaussian distributions
$$
f_{Y}(y) = \big(1 - \sum_{i=1}^{q-1} \delta_i \big) \, e^{-y^2/(2\sigma_0^2)}/(\sqrt{2\pi}\sigma_0)
+ \sum_{i=1}^{q-1} \delta_i \, e^{-y^2/(2\sigma_i^2)}/(\sqrt{2\pi}\sigma_i),
$$
where $\sigma_0 > 0$, $\sigma_i>0$, $\delta_i > 0$, for $i=1,2,\ldots, q-1$ and $\sum_{i=1}^{q-1} \delta_i < 1$.
Letting $\mathbf{\bar \delta} = (\delta_1, \delta_2, \ldots, \delta_{q-1})$,
for small $\norm{\bar \delta}_2$ we can
write the differential entropy of $Y$ as
$$
\h(Y) = \frac{1}{2} \log_2(2\pi e \sigma_0^2) +
\sum_{i=1}^{q-1} \delta_i(\sigma_i^2/\sigma_0^2 - 1)/(2\ln2) + O(\norm{\bar \delta}^2_2).
$$
\end{lem}

\begin{proof}
The proof is given in the Appendix~\ref{apd:b_h}.
\end{proof}

A direct application of \Lref{lem:b_h} results in the following corollary, which gives approximations for  two mutual information expressions. We use these approximations to compute the performance of binary and ternary SPDF schemes.
\begin{cor}
\label{cor:I}
Consider the $q$-ary random variable $Q$ with the distribution
$\Pr\left\{ Q = i \right\} = \delta_i$, for $i=0,1,\ldots q-1$,
where $\sum_{i=0}^{q-1} \delta_i =1$. We set
the distribution of the random variable $W$ to be 
$W | \{ Q = i \} \sim \cN(0, \sigma_i^2)$, for $i=0,1,\ldots, q-1$ and 
we define $Y = W + Z$ where $Z$ is statistically independent of $Q$ and
$W$ and has distribution $\cN(0,\sigma_Z^2)$. Letting
$\bar \delta = (\delta_1, \delta_2, \ldots, \delta_{q-1})$, then for a 
small $\norm{\bar \delta}_2$ we have
\begin{align*}
\I(Q ; W) =& \sum_{i=1}^{q-1}  \frac{\delta_i}{2\ln2} \left( \left(\frac{\sigma_i^2}{\sigma_0^2} - 1\right) -
\ln \left( \frac{\sigma_i^2}{\sigma_0^2} \right) \right) +
O(\norm{ \bar \delta }^2_2) \\
\I(W; Y | Q) =& \sum_{i=1}^{q-1} \frac{\delta_i}{2} \log_2 \left(
1 + \frac{\sigma_i^2}{\sigma_Z^2} \right).
\end{align*}
\end{cor}

\subsection{Symmetric diamond network with two relays}
\label{subsec:symmetric}

In this subsection we consider a network with $N=2$ relays, and symmetric gains of the form $g_1=g_2=g$ and $h_1=h_2=h$. 

The following theorem provides an asymptotic characterization of the maximum rate achievable with the Binary SPDF scheme.

\begin{thm}
\label{thm:asymptBSPDF}
For a symmetric diamond relay network 
assume that the channel gains from the source to relays are equal 
to $\sqrt{g}$ and channel gains from relays to the destination are
equal to $\sqrt{h}$. 
The rate $h\cdot\RBinary(g/h) + O(h^2)$ bits per network
use is achievable, where $\RBinary(g/h)$ is the solution of the following optimization problem.
\begin{align*}
\RBinary(g/h) =  \sup_{\beta > 0} &\  \Big( R_1 + R_2 \Big)
\\
\text{sub}&\text{ject to: } \, 
\\
R_1 & < 
\frac{1}{2\ln2} \left( g/h - \beta \sqrt{g/h} \ln \left[
1 + \frac{\beta}{\sqrt{g/h}} \right] \right)
\\
R_1 & <
\frac{1}{2\ln2} \left(
\frac{ 2 ( 2\sqrt{g/h} + \beta )}{\sqrt{g/h} + \beta} -
\beta \sqrt{g/h} \ln \left[1+
\frac{2( 2\sqrt{g/h} + \beta )}{ \beta ( g/h + \beta \sqrt{g/h}) }
\right] \right) 
\\
R_2 & < \frac{\beta \sqrt{g/h} }{2 \ln 2} \ln \left(1 + \frac{4\sqrt{g/h}}{\beta(2+g/h+\beta \sqrt{g/h})} \right).
\end{align*}
\end{thm}

\Tref{thm:N_relay} in the next subsection extends \Tref{thm:asymptBSPDF} for $N$-relay symmetric parallel networks. We omit the proof of \Tref{thm:asymptBSPDF} since it is a special case of \Tref{thm:N_relay}.

\begin{figure}[tb]


\centering
\includegraphics[width=\plotsize\textwidth]{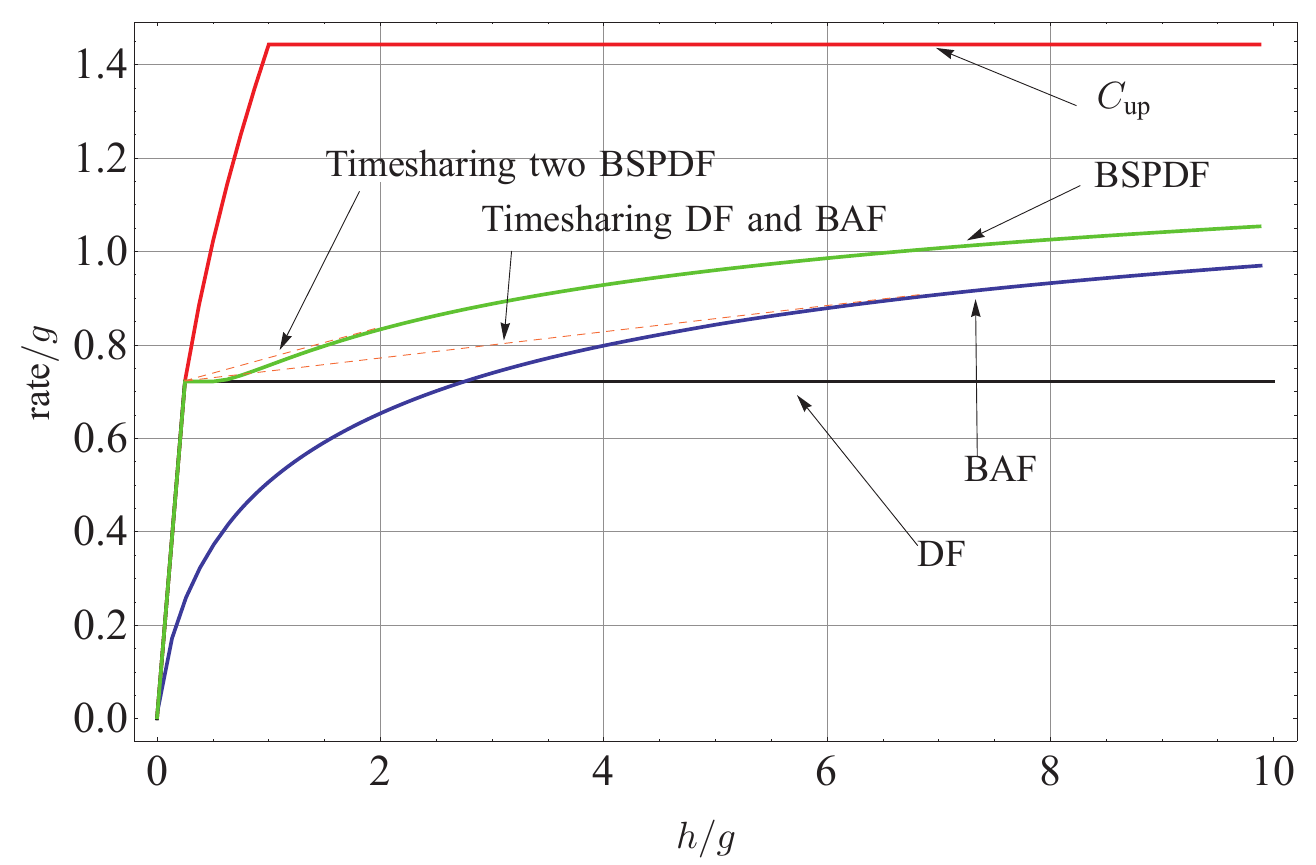}

\caption{The upper bound  and achievable rates of decode-and-forward (DF), bursty-amplify-and-forward (BAF), and binary-superposition-partially-decode-and-forward (BSPDF)  communication schemes in the symmetric diamond relay network. $\sqrt{g}$ is the channel gain from the source to relays and $\sqrt{h}$ is the channel gain from relays to the destination. 
}
\label{fig:BSPDF}

\end{figure}

We have solved the optimization for $\RBinary(g/h)$ numerically for a range of
$g/h$. The results are plotted in \Fref{fig:BSPDF} together with the achievable rates of DF, BAF, and the cut-set upper bound. Also plotted in the figure are the rates achievable by timesharing DF with BAF, and by timesharing two BSPDF schemes with different parameters (refer to Section \ref{ssec:timesharing} for an explanation of the rates achievable by timesharing different schemes). The result of timesharing between DF and BAF in \Fref{fig:BSPDF} appears as a segment of a line that starts at the point $h/g = 1/4$ and $\text{rate}/g = 1/(2 \ln 2)$ and is tangent to the curve of the achievable rate of BAF \cite{S01}. Similarly, the timesharing of two BSPDF schemes of different parameters results in a segment of a line that starts at the same point $h/g = 1/4$ and $\text{rate}/g = 1/(2 \ln 2)$ and it is tangent to the curve of the achievable rate of BSPDF. 

It follows from \Fref{fig:BSPDF} that BSPDF outperforms both DF and BAF in the low SNR regime. It also follows from the figure that for certain range of $h/g$ timesharing different BSPDF schemes results in higher rates. It is interesting to note that for $h/g\le 1/4$ the achievable rates of DF and BSPDF match the cut-set upper bound. Therefore, both of these schemes are optimal in the asymptotic low SNR regime for the range $0 \le h/g \le 1/4$.

\begin{figure}[tb]


\centering
\includegraphics[width=\plotsize\textwidth]{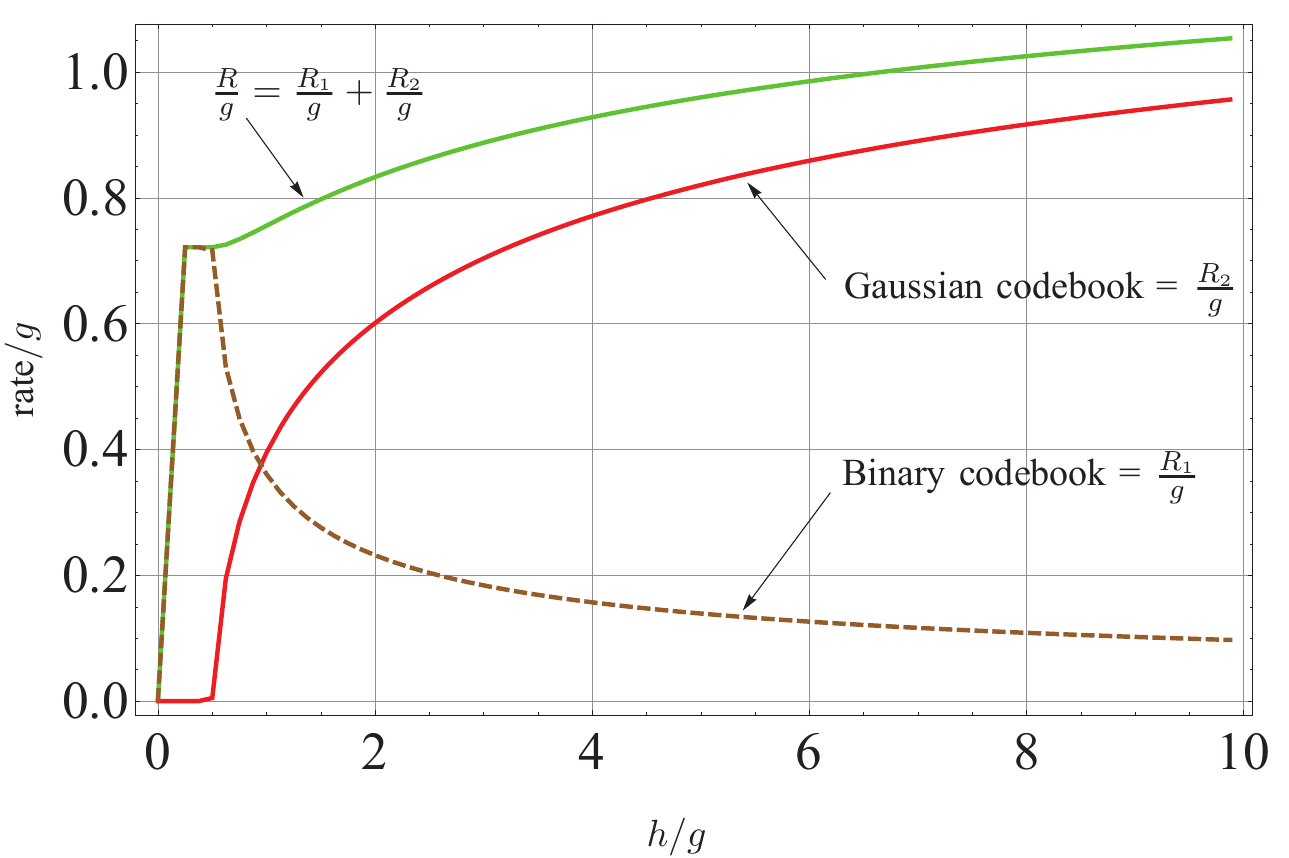}

\caption{Contribution of the binary codebook rate $R_1$ and the Gaussian codebook rate $R_2$ to the total communication rate $R=R_1+R_2$ in the Binary SPDF scheme for a symmetric 2-relay Gaussian diamond network.$\sqrt{g}$ is the channel gain from the source to relays and $\sqrt{h}$ is the channel gain from relays to the destination.}
\label{fig:bin_gaus_rate}

\end{figure}

We show in \Fref{fig:bin_gaus_rate} the contribution of the binary codebook rate $R_1$ and the Gaussian codebook rate $R_2$ to the total communication rate $R=R_1+R_2$ in the BSPDF scheme. We observe that for small values of $h/g$ when DF is optimal, BSPDF performs full decoding at the relays by assigning zero rate to the Gaussian codebook. As $h/g$ increases, it becomes beneficial to start conveying some rate in the Gaussian codebook. For $h/g \gg 1$ most of the rate is conveyed by the Gaussian codebook.

The fact that timesharing BSPDF schemes with different power levels
can result in a scheme with higher achievable rate suggests that a 
more flexible coding scheme that uses more than two power levels
to select the codeword symbols in the $\cC_{2,i_1}$ codebook, and that uses more amplification factors at the relays, should outperform the BSPDF scheme. 

To that effect we consider the Ternary SPDF scheme as described in \Tref{thm:TSPDF} and characterize its performance in the asymptotic low SNR regime in the following theorem.

\begin{thm}
\label{thm:asymptTSPDF}
For a symmetric diamond network with two relays 
assume that the channel gains from the source to relays are equal 
to $\sqrt{g}$ and channel gains from relays to the destination are
equal to $\sqrt{h}$.
The rate $h \cdot \RTernary(g/h) + O(h^2)$ is achievable 
where $\RTernary(g/h)$ is the solution to the following optimization
\begin{align*}
  \RTernary(g/h)  = \sup_{\beta_1, \beta_2, \gamma_1, \gamma_2, \kappa_1, \kappa_2} 
\ &\Big( R_1 + R_2 \Big)
\\
 \text{ subject to: } &
\\
 R_1 & < 
\frac{\sqrt{g/h }}{2 \ln 2} \Big( \beta _1 \left(\gamma_1 g/h - \ln\left[1+ \gamma_1 g/h    \right] \right)+ 
\beta_2 \left( \gamma_2 g/h   -\ln\left[1+ \gamma_2 g/h  \right]\right) \Big) 
\\
 R_1 & < 
\frac{\sqrt{g/h}}{2 \ln 2} \Big( \beta _1 \left(\left(2+ 4 \gamma_1 g/h   -\ln\left[1+\left(2+4 \gamma_1 g/h  \right) \kappa_1\right] \right) \kappa_1\right)+
\\
& \quad\quad\quad\quad\quad
\beta _2 \left(\left(2+4 \gamma _2 g/h  \right) \kappa _2-\ln\left[1+\left(2+4 \gamma_2 g/h  \right) \kappa _2\right]\right) \Big)  
\\
 R_2  & < 
\frac{\sqrt{g/h}}{2 \ln 2} \left(\ln\left[1+\frac{4 (\gamma_1 g/h)   \kappa _1}{1+2 \kappa _1}\right] \beta _1+\ln\left[1+\frac{4 (\gamma_2 g/h)  \kappa _2}{1+2 \kappa _2}\right] \beta _2\right)
\\
 0 & < \beta_1 ,\beta_2  
\\
 0 & \le \gamma_1, \gamma_2, \kappa_1, \kappa_2  
\\
 1&>\sqrt{g/h }  \left(\beta _1 \left(1+ \gamma_1 g/h \right) \kappa _1+\beta _2 \left(1+ \gamma_2 g/h \right) \kappa _2\right)
\\
 1&> \sqrt{g/h} \left(\beta _1 \gamma _1+\beta _2 \gamma _2\right).
\end{align*}

\end{thm}
\begin{proof}
The proof is given in Appendix \ref{app:asympTSPDFproof}.
\end{proof}

\begin{figure}[tb]


\centering
\includegraphics[width=\plotsize\textwidth]{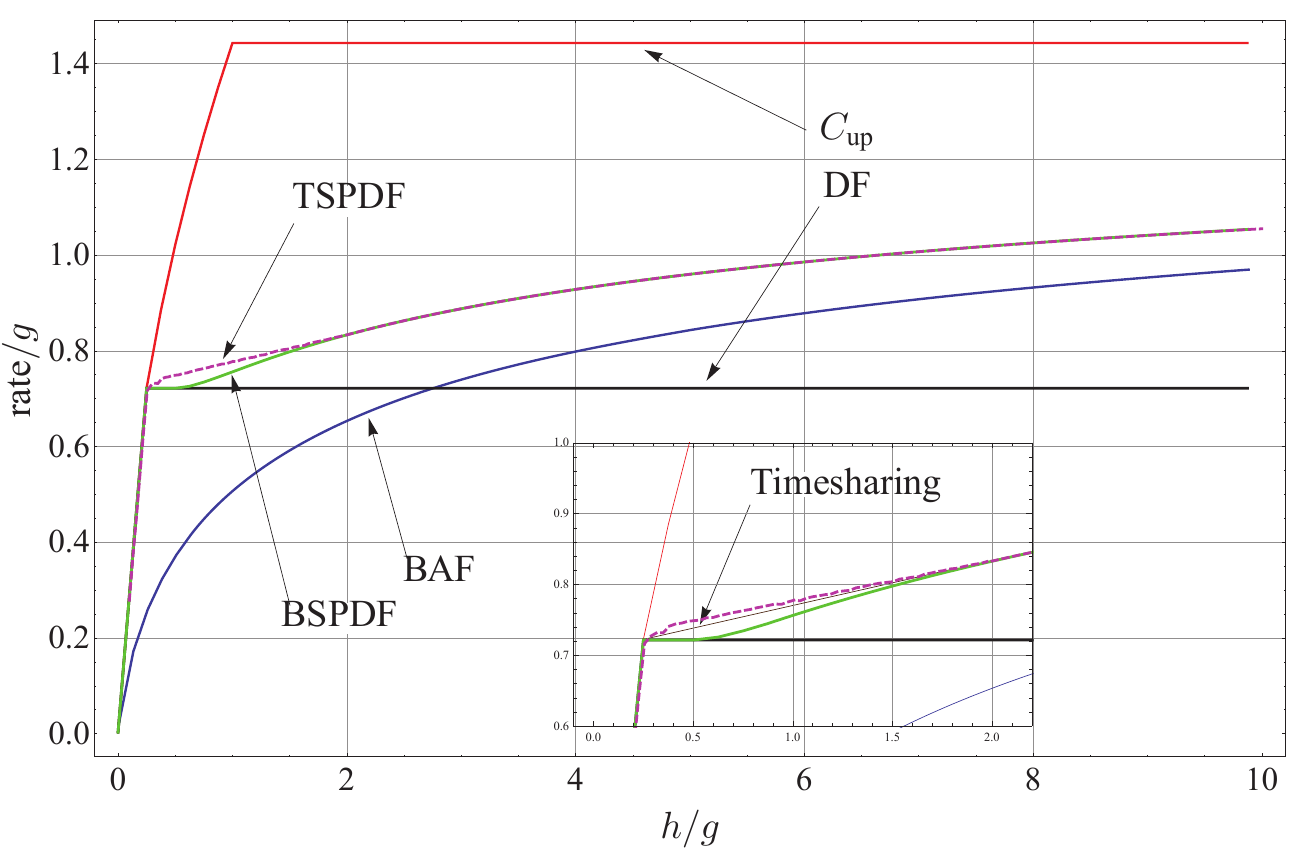}

\caption{The upper bound and achievable rates for decode-and-forward (DF), bursty-amplify-and-forward (BAF), Binary SPDF (BSPDF) and Ternary SPDF (TSPDF) communication schemes in the symmetric diamond network with two relays. $\sqrt{g}$ is the channel gain from the source to relays and $\sqrt{h}$ is the channel gain from relays to the destination. The inset shows an expanded view comparing the achievable rates of TSPDF to timeshared BSPDF.}
\label{fig:TSPDF}

\end{figure}

Solving the optimization problem of \Tref{thm:asymptTSPDF} is difficult.
We have been able to find local optima for various values of $h/g$ which are plotted
in \Fref{fig:TSPDF}. Notice that any local optimum in the optimization
of \Tref{thm:asymptTSPDF} also results in an achievable rate for the diamond relay network.
These numerical results show that TSPDF outperforms timeshared BSPDF and other known relaying 
schemes for symmetric diamond networks with two relays in the asymptotic low SNR regime.


\subsection{Symmetric parallel network with N relays}
\label{sec:N_relays}

So far we have only considered diamond networks with two relays.
In this section we consider a symmetric parallel relay network with $N > 2$
relays. We assume a symmetric network where the channel gains from source to relays are equal to
$\sqrt{g}$ and the channel gains from relays to destination are equal
to $\sqrt{h}$. Furthermore, we assume that $g, h \ll 1$ with a fixed ratio $g/h$.


In the following theorem we characterize the achievable rate of the Binary SPDF scheme in the asymptotic low SNR regime. We use \Tref{thm:BSPDF} and approximate the mutual information expressions for small $g$ and $h$ using \Cref{cor:I}. 

\begin{thm}
\label{thm:N_relay}
The rate $h \cdot \RBinary(N, g/h) + O(h^2)$ is achievable
in the symmetric $N$-relay Gaussian parallel network, where $\sqrt{g}$ is the channel gain from the source to relays and $\sqrt{h}$ is the channel gain from relays to the destination and $\RBinary(N,g/h)$ is the solution of the following optimization
\begin{align*}
\RBinary(N,g/h)  = \sup_{\beta > 0} & \Big( R_1+R_2 \Big) \\
 \text{sub}&\text{ject to:} \\
 R_1 & < 
\frac{1}{2\ln 2} \left( g/h - \beta  \sqrt{g/h }  \ln \left[1 + \frac{\sqrt{g/h} }{\beta}\right] \right) \\
 R_1 & < 
\frac{1}{2\ln 2} \left( \frac{N \left(N \sqrt{g/h }+\beta \right)}{\sqrt{g/h }+\beta } - \beta\sqrt{g/h }  \ln \left[ 1 + 
\frac{N(N \sqrt{g/h} + \beta) }
{\beta(g/h + \beta \sqrt{g/h})}
\right] \right) \\
 R_2 & < 
\frac{ \beta\sqrt{g/h }  }{2 \ln 2} \ln \left[1+\frac{ N^2 \sqrt{g/h }}{\beta  \left(N+g/h + \beta\sqrt{g/h }  \right)}\right].
\end{align*}
\end{thm}

\begin{proof}
The proof of the theorem is given in Appendix~\ref{apd:N_relay}.
\end{proof}

\begin{figure}[tb]


\centering
\includegraphics[width=\plotsize\textwidth]{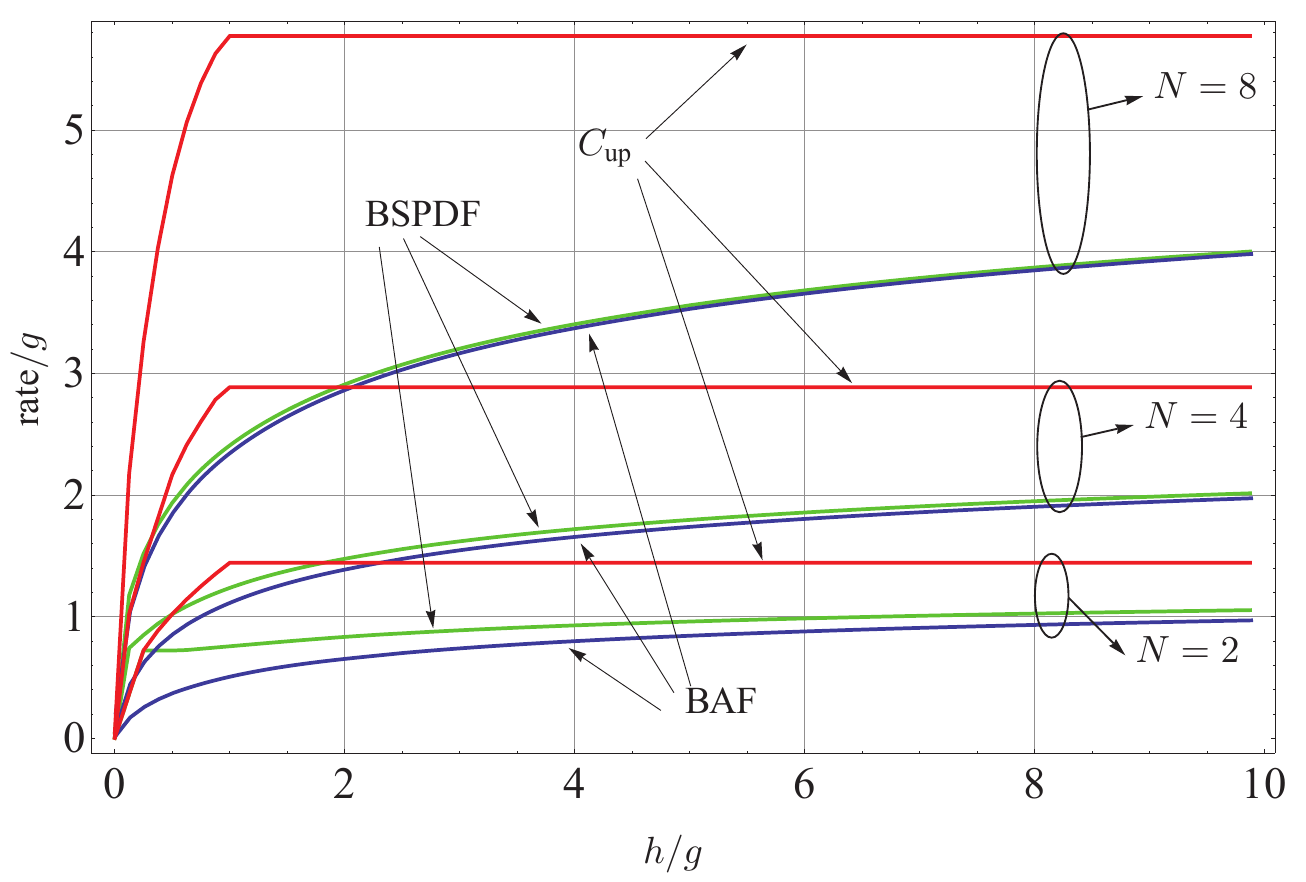}

\caption{Cut-set upper bound and achievable rates of the bursty-amplify-and-forward (BAF) and
the Binary SPDF (BSPDF) communication schemes in the symmetric Gaussian parallel relay network with  $N=2,4$ and $8$ relays.
$\sqrt{g}$ is the channel gain from the source to relays and $\sqrt{h}$ is the channel gain from relays to the destination.}
\label{fig:N_relays}

\end{figure}

We numerically computed the achievable rates of BAF and 
BSPDF, and the results are given in \Fref{fig:N_relays} for
networks with $N=2,4,8$ relays. \Fref{fig:N_relays} also shows the capacity upper bound of \Pref{pro:cut_setN} for these networks.  We see in this figure that as the number of relays increases, the gap between the rates of BSPDF and BAF shrinks. As $N$ increases for a fixed $h/g$ ratio, the links between the relays and the destination benefit as a whole from a larger beamforming gain while the links from the source to each relay remain invariant. Exploiting the benefits of larger $N$ requires increasing the data rate $R=R_1+R_2$, which makes it harder for the relays to decode the lower level codewords $\bu_{1,i_1}$. Successful decoding of the lower level messages requires in turn to reduce $R_1$. In addition, BAF can be thought of as a special case of BSPDF with $R_1=0$. As a result, as $N$ increases the performance of BSPDF approaches that of BAF.


\subsection{Asymmetric diamond network with two relays}
\label{sec:asym}

In the previous subsections we considered symmetric parallel relay networks where $g_1=\ldots=g_N$ and $h_1=\ldots=h_N$. These symmetric networks model scenarios where all relays are at similar distances to the source and destination. In this section we consider a diamond relay network with $N=2$ relays, where one of the relays is close to the source, while the other is close to the destination. To simplify the analysis we assume $g_1=h_2=g$ and $g_2=h_1=h$, and as before, consider the regime of $g,h \to 0$ with a fixed ratio $h/g$ (see \Fref{fig:asym_channel}). WLOG we assume $g \le h$. 


\begin{figure}[tb]

\centering
\includegraphics[height=0.30\textwidth]{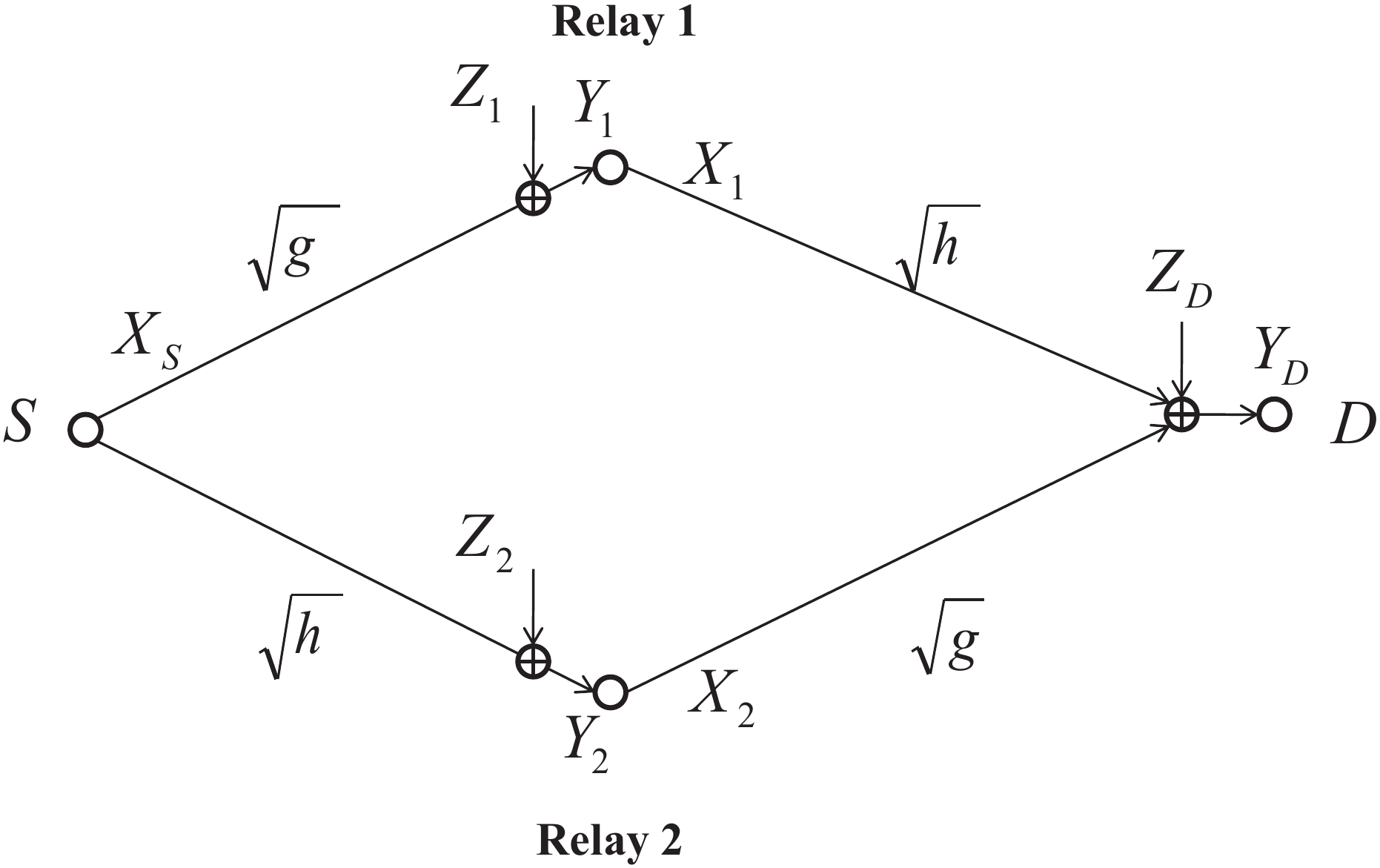}

\caption{Communication model for the asymmetric 2-relay diamond network.}
\label{fig:asym_channel}

\end{figure}

As we did in the previous subsections, we analyze the achievable rates with the Binary SPDF communication scheme in the low SNR limit for this network, and compare them to the rates that can be obtained with DF and BAF. In the asymmetric network setting we consider the extension of the BSPDF scheme given in \Cref{cor:BSPDF2}, where different relays decode different message levels. The function $f(i)$ specifies the maximum message level decoded by relay $i$. We first consider the case where both relays decode only the first message layer ($f(1)=f(2)=1$). Using \Tref{thm:BSPDF} or specializing \Cref{cor:BSPDF2} to this choice of $f$, and approximating the mutual information expressions for small $g$ and $h$ using \Cref{cor:I} we obtain the following result.

\begin{thm}
\label{thm:asym_BSPDF}
WLOG we assume $g \le h$. The BSPDF(f(1)=1, f(2)=1) communication scheme in the asymmetric 2-relay diamond network achieves a
rate $h\cdot R_{\text{BSPDF}(1,1)}(g/h) + O(h^2)$, where 
\begin{align*}
R_{\text{BSPDF}(1,1)}(g/h)  = \sup_{\kappa_1, \kappa_2, \beta} \Big( R_1 & +R_2  \Big) \\
 \text{subject to:} & \\
 R_1 & < 
\frac{1}{2 \ln 2} \left( g/h +\sqrt{g/h } \beta  \ln\left[\frac{\beta }{\sqrt{g/h }+\beta }\right] \right) \\
 R_1 & < 
\frac{\beta \sqrt{g/h } }{2\ln 2}\Bigg(-  \ln \left[1+ \kappa _1 + \kappa_2 g/h  + \frac{\sqrt{g/h } \left(\sqrt{\kappa _1}+\sqrt{\kappa _2}\right){}^2}{\beta }\right]+ \\
& \quad\quad\quad
\kappa _1 + \kappa_2 g/h  + \frac{\sqrt{g/h } \left(\sqrt{\kappa _1}+\sqrt{\kappa _2}\right){}^2}{\beta } \Bigg)
 \\
 R_2 & < 
\frac{\beta \sqrt{g/h } }{2 \ln 2}  \ln\left[1+\frac{\sqrt{g/h } \left(\sqrt{\kappa _1}+\sqrt{\kappa _2}\right){}^2}{\beta  \left(1+\kappa _1+ \kappa_2 g/h \right)}\right] \\
& \kappa_1(g/h + \beta \sqrt{g/h} ) < 1 , \ \ \
 \kappa_2(1 + \beta \sqrt{g/h}) < 1 \\
 & \kappa_1 \ge 0, \kappa_2 \ge 0, \beta > 0.
\end{align*}

\end{thm}
\begin{proof}
The proof of the theorem is given in Appendix~\ref{apd:asym_BSPDF}.
\end{proof}

Due to the asymmetry in the network, it may be beneficial to consider a variation of the BSPDF scheme in which the relay with larger source-relay channel gain attempts to completely decode the source messages. This corresponds to choosing $f(1)=1, f(2)=2$. As a baseline for comparison we consider a communication scheme based on BAF, which we call BAF+DF where the relay with larger source-relay channel gain decodes and forwards the source message (as in DF) while the other relay amplifies the received signal during the burst interval (as in BAF). The following proposition characterizes the achievable rates with BAF+DF.

\begin{prop}
WLOG we assume $g\le h$.
The BAF+DF relaying scheme can achieve a rate $h \cdot R_{\text{BAF+DF}}(g/h) + O(h^2)$,
where
\begin{align*}
R_{\text{BAF+DF}}(g/h)  = \sup_{\beta, \kappa}  \ \ R \ &\\
 \text{subject to:} & \\
 R & < 
\frac{ \beta \sqrt{g/h }  }{2\ln 2} \ln \left[1 + \frac{1}{ \beta \sqrt{g/h }  }\right] \\
 R & < 
\frac{ \beta  \sqrt{g/h } }{2 \ln 2} \ln \left[1 + \frac{g/h \left(1+
\sqrt{\kappa }\right){}^2}{\beta \sqrt{g/h }  \left(1 + 
\kappa \right)}\right] \\
 \beta &> 0, \ \kappa \ge 0, \ \kappa( \beta \sqrt{g/h} + 
g/h) < 1.
\end{align*}

\end{prop}

We now consider the BSPDF($f(1)=1,f(2)=2$) scheme. Using \Cref{cor:BSPDF2} and approximating the mutual information expressions in the low SNR limit with \Cref{cor:I} we we obtain the following result.

\begin{thm}
\label{thm:BSPDF+DF}
WLOG we assume $g \le h$.
The {\it BSPDF}($f(1)=1,f(2)=2$) relaying scheme can achieve a rate $h \cdot R_{\text{BSPDF(1,2)}}(g/h) + O(h^2)$, where 
\begin{align*}
R_{\text{BSPDF(1,2)}}(g/h)  = \sup_{\beta, \kappa} \Big( 
R_1 +& R_2 \Big) \\
 \text{subject to:} & \\
 R_1 & < 
\frac{1}{2\ln 2} \left( g/h + \beta  \sqrt{g/h }  \ln\left[\frac{\beta }{\sqrt{g/h }+\beta }\right] \right) \\
R_1 & < \frac{1}{2 \ln 2}
\Big((g/h)(1+
\sqrt{\kappa})^2 +
\kappa - 
\beta \sqrt{g/h} \ln \Big[
1 + 
\kappa + 
\frac{\sqrt{g/h} (1 + \sqrt{\kappa})^2}{\beta} 
\Big]
\Big)  \\
 R_2 & <  
\frac{\beta \sqrt{g/h }  }{2 \ln 2} \ln\left[1 + \frac{1}{\beta \sqrt{g/h }  }\right] \\
 R_2 & < 
\frac{\beta \sqrt{g/h }  }{2 \ln 2} \ln\left[1 + \frac{(g/h)\left(1+
\sqrt{\kappa }\right){}^2}{\beta  \sqrt{g/h } \left(1 + 
\kappa \right)}\right] \\
 \beta &> 0, \ \kappa \ge 0, \ \kappa(\beta\sqrt{g/h} + g/h ) < 1.
\end{align*}
\end{thm}
\begin{proof}
The proof of the theorem is given in Appendix~\ref{apd:BSPDF+DF}.
\end{proof}

\begin{figure}[tb]


\centering
\includegraphics[width=\plotsize\textwidth]{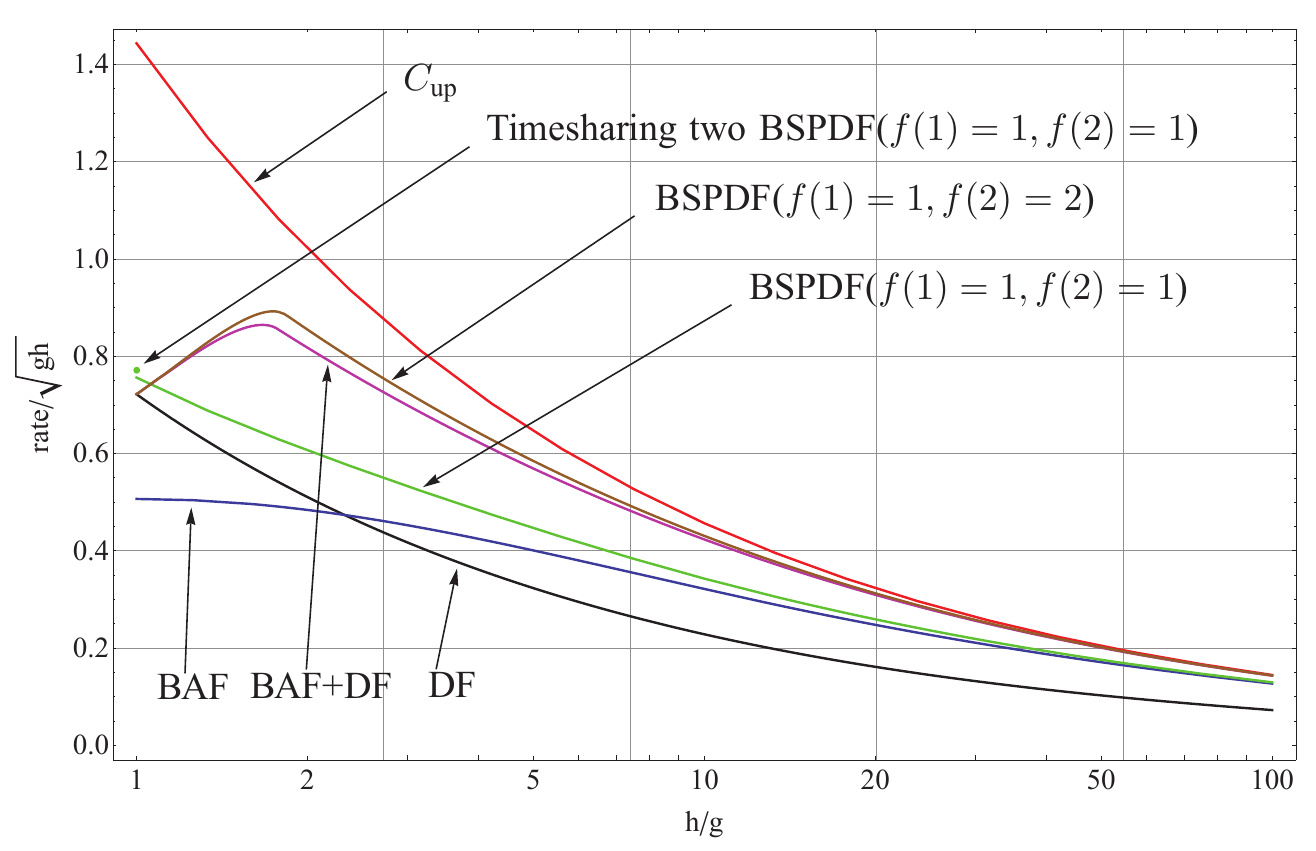}

\caption{Upper bound and achievable rates for various communications schemes in the asymmetric Gaussian diamond network. The schemes considered are decode-and-forward (DF), bursty-amplify-and-forward (BAF), bursty-amplify-and forward at relay 1 and decode-and-forward at relay 2 (BAF+DF), Binary SPDF (BSPDF($f(1)=1,f(2)=1$)), and the extension of BSPDF where relay 2 performs decode and forward by decoding both level codewords (BSPDF($f(1)=1,f(2)=2$)). The communication model is given in \Fref{fig:asym_channel}. Timesharing two BSPDF schemes of different parameters achieves better performance as depicted for the case $h=g$. The timesharing point is estimated from \Fref{fig:BSPDF}.
}
\label{fig:nonsymmetric}

\end{figure}

We numerically computed the achievable rates of DF, BAF, BSPDF($f(1)=1,f(2)=1$), BAF+DF, and BSPDF($f(1)=1, f(2)=2$) as a function of $h/g$ and the results are given in \Fref{fig:nonsymmetric}. Also shown in \Fref{fig:nonsymmetric} is the achievable rate resulting from timesharing BSPDF with itself for $h/g=1$, which can be obtained from the BSPDF timesharing curve of \Fref{fig:BSPDF} for the symmetric diamond relay network with $N=2$. We see that for $h/g=1$ and possibly for other values of $h/g$ (not shown), timesharing BSPDF achieves a larger rate than any of the pure schemes. 

\Fref{fig:nonsymmetric} shows that BSPDF($f(1)=1,f(2)=1$), BSPDF($f(1)=1,f(2)=2$) and timeshared BSPDF($f(1)=1,f(2)=1$) provide the largest achievable rates of all the schemes considered.

One may possibly obtain larger rates by increasing the alphabet size of the first level codebook as was done in \Tref{thm:asymptTSPDF} for the symmetric diamond relay network. The analysis should be similar to the one in \Tref{thm:asymptTSPDF} and we do not consider it in this paper.


\section{Bounds on the Minimum Energy-Per-Bit}
\label{sec:ebno}

In this section, we use the results of Sections~\ref{sec:achievable}-\ref{sec:analysis} to compute upper and lower bounds on the minimum energy-per-bit for the symmetric diamond relay network. An extension of Theorem~1 in \cite{EMZ06} shows that the minimum energy-per-bit of the diamond relay network can be characterized as
\begin{equation}
\label{eq:ebno_general}
\cE_b = \inf_{\gamma_1, \gamma_2 > 0} \inf_{P > 0 } \frac{(1+\gamma_1+\gamma_2)P}{C(P, \gamma_1 P, \gamma_2 P)},
\end{equation}
where $C(P,\gamma_1 P, \gamma_2 P)$ represents the capacity of the diamond relay network with average power constraint $P$, $\gamma_1 P$ and $\gamma_2 P$ on the source, relay~$1$ and relay~$2$ nodes respectively. For a symmetric diamond relay network, we know that by symmetry $C(P, \gamma_1 P, \gamma_2 P) = C(P, \gamma_2 P, \gamma_1 P)$. Thus, by timesharing we have
$$
C(P, (\gamma_1+\gamma_2)P/2, (\gamma_1+\gamma_2)P/2) \ge C(P, \gamma_1 P, \gamma_2 P).
$$
So for any $\gamma_1, \gamma_2 > 0$ if we set 
$\gamma = (\gamma_1 + \gamma_2)/ 2$ we can write
$$
\frac{(1+\gamma_1+\gamma_2)P}{C(P, \gamma_1 P, \gamma_2 P)}
\ge
\frac{(1+\gamma_1+\gamma_2)P}{C(P, (\gamma_1+\gamma_2)P/2, (\gamma_1+\gamma_2)P/2)}
=
\frac{(1+2 \gamma)P}{C(P, \gamma P, \gamma P)},
$$
and the minimum energy-per-bit for the symmetric diamond relay network becomes
\begin{equation}
\label{eq:enbo_1}
\cE_b = \inf_{\gamma > 0} \inf_{P > 0 } \frac{(1+ 2\gamma)P}
{C(P, \gamma P, \gamma P)}.
\end{equation}
The capacity function $C(P, \gamma P, \gamma P)$  satisfies the conditions of Lemma~1 in \cite{EMZ06} and therefore 
$\frac{(1+ 2\gamma)P} {C(P, \gamma P, \gamma P)}$ 
is a non-decreasing function in $P$, for all $P > 0$ and $\gamma > 0$. This implies that we can replace $\inf_{P>0}$ by $\lim_{P\to 0}$ in \eq{eq:enbo_1} obtaining
\begin{equation}
\label{eq:ebno}
\cE_b = \inf_{\gamma > 0} \lim_{P \to 0 } \frac{(1+ 2\gamma)P}
{C(P, \gamma P, \gamma P)}.
\end{equation}

A lower bound on $\cE_b$ can be obtained by upper bounding $C(P, \gamma P, \gamma P)$ in \eq{eq:ebno} using the cut-set upper bound of \Pref{prop:cut_set}. 
Since in Section~\ref{sec:upper_bound} we have assumed the average power constrains on source and relays are one and variances of the additive Gaussian noise at relays and destination are one, we need to scale the channel gains appropriately in \Pref{prop:cut_set}. We set the normalized channel gains to be $\tilde g = P g /N_0$ and $\tilde h = \gamma P h /N_0$, i.e. we assume $\sqrt{\tilde g}$ as the channel gain from source to relays and $\sqrt{\tilde h}$ as the channel gain from relays to destination in \Pref{prop:cut_set}, obtaining
\begin{equation}
\cE_b \ge \inf_{\gamma > 0} \lim_{P\rightarrow 0} \frac{(2\gamma+1)P}{\Cupper} = \left\{ 
\begin{array}{ll}
\frac{(g + 2h) N_0 \ln 2  }{gh}        & ; h/g \le  1/2  \\[1.4ex]
\frac{\sqrt{8} N_0 \ln 2 }{\sqrt{gh}} & ; 1/2 < h/g \le 2 \\[1.4ex]
\frac{(h + 2g) N_0 \ln 2 }{gh }        & ; 2 < h/g  
\end{array}
\right..
\end{equation}
A detailed derivation of this lower bound is given in Appendix~\ref{apd:ebno_upper_bound}.

In addition, we can find an upper bound on the minimum energy-per-bit of the diamond relay network by lower bounding the capacity in \eq{eq:ebno} with an achievable rate of any communication scheme. Since in Section~\ref{sec:achievable} and Section~\ref{sec:analysis} we have also assumed that the average power constraints at the source and relays are one and the variance of the additive Gaussian noise at relays and destination is one, we use the same normalization  for the channel gains to compute the achievable rates, i.e. we  set the normalized channel gains to $\tilde g = P g /N_0$ and $\tilde h = \gamma P h /N_0$.

If we use DF the achievable rate is given by \eq{eq:rdf3}, and \eq{eq:ebno} results in
\begin{equation}
\cE_b \le \inf_{\gamma > 0} \lim_{P\rightarrow 0} \frac{(2\gamma + 1)P }{\RDF} = 
\frac{(g + 2h)  N_0 \ln 2}{g h}.
\end{equation}
If we use the achievable rate of BAF given in \eq{eq:ebno} we get the following upper bound on the minimum energy-per-bit of the symmetric diamond relay network
\begin{equation}
\cE_b \le \inf_{\gamma>0} \lim_{P\rightarrow 0} \frac{(2\gamma+1)P}{\RBAF} = 
\inf_{\beta>0, \gamma>0}
\frac{2 (1+2\gamma) N_0  \ln 2}
{
\beta \sqrt{\gamma g h}
\ln \left[
1 + 
\frac{4 \sqrt{ \gamma g h}}
{\beta ( 2 \gamma h + g + \beta \sqrt{ \gamma g h} )}
\right]
}.
\end{equation}

Using the achievable rate of the BSPDF scheme we can upper bound the minimum energy-per-bit of the symmetric diamond relay network by 
\begin{align}
\cE_b  \le  \inf_{\gamma>0} \frac{(2\gamma+1)N_0}{\gamma h \RBinary(g/ (\gamma h))}. 
\end{align}

\begin{figure}[tb]


\centering
\includegraphics[width=\plotsize\textwidth]{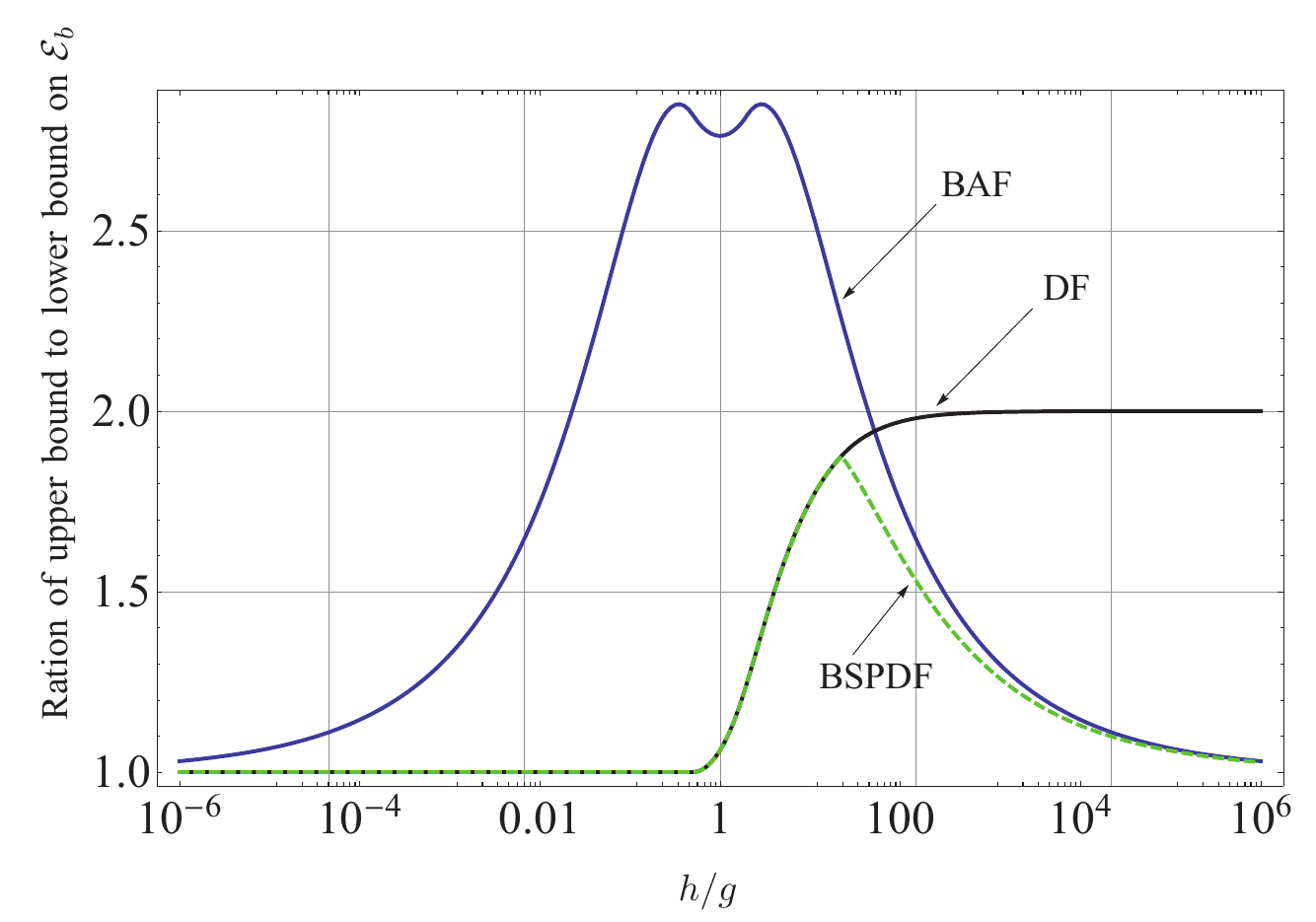}

\caption{Ratio of the upper bound to the lower bound on the minimum energy-per-bit for bursty-amplify-and-forward (BAF), decode-and-forward (DF), and Binary SPDF (BSPDF) communication schemes in the symmetric diamond relay network. The worst ratio of the upper bound to the lower bound on the minimum energy-per-bit for BAF is around 2.85, for DF is 2.0 and for BSPDF is 1.87.}
\label{fig:ebno_ratio}

\end{figure}

We have plotted the ratio of the upper bound to the lower bound on the minimum energy-per-bit of the symmetric diamond relay network for different ratios of $h/g$ in~\Fref{fig:ebno_ratio}. Note that the ratio between the upper and lower bounds on $\cE_b$ depends on $(g,h)$ through $h/g$. 

For BAF, the ratio of the upper bound to the lower bound on $\cE_b$ goes to one if the channel gains from relays to the destination are much larger or much smaller than the channel gains from the source to relays. The worst-case ratio of the upper bound to the lower bound on the minimum energy-per-bit for bursty-amplify-and-forward is approximately 2.85. 

In the DF communication scheme the upper bound and the lower bound for minimum energy-per-bit become identical when $h/g \le 1/2$; however, in the worst case the ratio of the upper bound to the lower bound on the minimum energy-per-bit approaches $2$ when  $h/g \rightarrow \infty$. 

Combining the achievable rates of BAF and DF results in a worst-case ratio of 1.94 between the upper and lower bounds on $\cE_b$. Thus, using previously known communication schemes, the minimum energy-per-bit for the symmetric diamond network can be characterized within a factor of 1.94.

The BSPDF communication scheme outperforms both the BAF and DF schemes with respect to the minimum energy-per-bit metric for any values of $h/g$. The worst-case ratio of the upper bound to the lower bound on the minimum energy-per-bit of the BSPDF scheme is $1.87$ and if $h/g \le 1/2$, as in DF, the ratio becomes equal to one. Thus, BSPDF allows to tighten the minimum energy-per-bit characterization for the symmetric diamond network to within a factor of 1.87.


\section{Conclusion}
\label{sec:conclusion}

We presented a new communication scheme, SPDF, for Gaussian $N$-relay  parallel networks
based on superposition coding and partial decoding at the relays. Superposition coding allows
to construct codebooks of different rates targeted to the decoding capabilities of the various relays. Partially decoding the source message at the relays allows to reduce noise amplification while relaying the message. The proposed scheme is a generalization of DF, AF, and BAF, and achieves rates at least as large as any of these schemes.

We have shown that by using low duty cycle codebooks and two codebook layers it is is possible to obtain strictly better performance than both DF and BAF in the low SNR regime. While most of the rate gains are readily obtained with the BSPDF scheme that uses binary symbols in the codewords of the first level codebook, some additional rate improvements results from either timesharing BSPDF schemes or increasing the alphabet size of the first level codebooks (e.g. TSPDF).

We presented a characterization of the achievable rates of BPSDF, TSPDF, and the general SPDF scheme in terms of single letter expressions. These expressions often require the computation of differential entropies of mixtures of Gaussian random variables, and as a result, do not lead in general to closed form solutions. For the low SNR asymptotic regime we provided closed form  formulas for the achievable rates of BSPDF and TSPDF, which for ease of presentation, we derived for channels with some symmetry properties.

Computation of the maximum achievable rates, even in the low SNR asymptotic regime, requires optimization over some design parameters. We evaluated the performance of BSPDF and TSPDF by numerically solving the corresponding optimization problems, and compared their performance to that of previously known schemes. We note that the requirement of optimizing over some design parameters is not unique to BSPDF or TSPDF, since the maximum achievable rate of BAF also requires solving an optimization problem which lacks a known closed form solution.

We showed that BSPDF and TSPDF achieve strictly larger rates than the previously known schemes in the low SNR regime for various types of parallel relay channels with some symmetry properties. The symmetry properties of these channels allowed us to present results for a wide range of channel gains in two-dimensional plots. We observed that as the number of relays $N$ grows, the performance gains of BSPDF over BAF decrease.

Finally we showed how our achievable rate characterization for BSPDF in the low SNR asymptotic limit can be used to obtain a minimum energy-per-bit upper bound. We used this new bound to obtain a characterization of the minimum energy-per-bit for the symmetric Gaussian diamond relay network within a constant multiplicative factor of 1.87, improving over a factor of 1.94 that results from using previously known communication schemes to derive the minimum energy-per-bit upper bound.

Since our analysis focused on the low SNR regime, our results do not allow us to improve the capacity characterizations within a constant factor or within a constant gap of Gaussian parallel relay networks. However, it may be possible to derive achievable rate lower bounds for BSPDF, TSPDF, or SPDF, for non-asymptotic SNR which may result in tighter than the currently known capacity characterizations within a constant multiplicative factor or additive gap. 


\appendices


\section{Proof of \Tref{thm:BSPDF}}
\label{app:proofBSPDF}

\Tref{thm:SPDF} requires $(U_1,\ldots, U_K)$ to be continuous random variables. We will use a continuous distribution for $(U_1, U_2)$ that emulates the discrete distribution for $B$ and the mixed distribution for $X_S|B$ given in the statement of the theorem.

In \Tref{thm:SPDF}, let $K=2$. We set $U_1 \sim \text{Uniform}[0,1]$, 
$U_2 | \{U_1 \in [0,1-\delta]\} \sim \cN(0,\ep)$ for some small $\ep > 0$ and $U_2 | \{U_1 \in (1-\delta, 1]\} \sim \cN(0,\sigma^2)$. Also let $f(1)=\ldots=f(N)=1$, and define
\begin{equation*}
w_i(U_1, V_i)  = \left\{
\begin{array}{ll}
0 & \text{ if } U_1 \in [0,1-\delta] \\
\sqrt{\kappa_i} V_i & \text{ if } U_1 \in (1-\delta, 1]
\end{array} \right.
\end{equation*}
for $i=1,\ldots,N$.

Assuming  $\delta \sigma^2 + (1-\delta)\ep < 1$ and
$\delta \kappa_i (g_i \sigma^2 +1) <1 $, then by \Tref{thm:SPDF} the rate $\tilde R = \tilde R_1 + \tilde R_2$ is achievable if
\begin{align}
\label{eq:cor_BSPDF_1}
\tilde R_1 & < \I(U_1; V_1), \\
\label{eq:cor_BSPDF_2}
\tilde R_2 & < \I(U_2; V_D| U_1), \\
\label{eq:cor_BSPDF_3}
\tilde R_1 + \tilde R_2 & < \I(U_1, U_2; V_D).
\end{align}

If we can show that
\begin{align}
\label{eq:lim1}
\lim_{\ep \to 0} \I(U_1; V_1) =& \I(B;Y_1), \\
\label{eq:lim2}
\lim_{\ep \to 0} \I(U_2; V_D| U_1) =& \I(X_S; Y_D| B),\\
\label{eq:lim3}
\lim_{\ep \to 0} \I(U_1, U_2; V_D) =& \I(B, X_S; Y_D),
\end{align}
then it would follow that for any rate $R=R_1+R_2$ that satisfies the conditions of \Tref{thm:BSPDF} we can find an $\ep>0$ small enough such that $R_1 \le \tilde R_1$ and $R_2 \le \tilde R_2$,
for some $\tilde R_1 $ and $\tilde R_2$ satisfying \eq{eq:cor_BSPDF_1}, \eq{eq:cor_BSPDF_2}, and \eq{eq:cor_BSPDF_3}. This, together with \Tref{thm:SPDF} would imply that the rate $R = R_1+R_2$ is achievable.

In order to prove \eq{eq:lim1} we will show that $\lim_{\ep\to 0}\h(V_1)=\h(Y_1)$. To that end, we note that 
$$
P_{V_1}(x) = \delta \frac{1}{\sqrt{2 \pi (g_1 \sigma^2+1)}} e^{-x^2/(2(g_1 \sigma^2+1))} + (1-\delta)\frac{1}{\sqrt{2 \pi (g_1\ep+1)}} e^{-x^2/(2(g_1\ep+1))}
$$ 
converges to $P_{Y_1}(x)$ as $\ep \to 0$ for every $x$. In addition, the continuity of the $\log_2(\cdot)$ function implies that $P_{V_1}(x) \log_2 (P_{V_1}(x))$ converges to $P_{Y_1}(x) \log_2 (P_{Y_1}(x))$ as $\ep \to 0$ for every $x$. 

Letting 
\begin{align*}
A=&\max\Big(\Big|\log_2 \frac{2\max(\delta, 1-\delta)}{\sqrt{2\pi}}\Big|,\Big|\log_2\frac{\min(\delta,1-\delta)}{\sqrt{2\pi(g_i \sigma^2+1)}}\Big|\Big), \\
g(x)=&\sqrt{\frac{2}{\pi}}\max(\delta,1-\delta)e^{-\frac{x^2}{2(g_1\sigma^2+1)}}\Big(A+\frac{x^2}{2 \ln 2}\Big),
\end{align*}
we have that $|P_{V_1}(x) \log_2 (P_{V_1}(x))| \le g(x)$ for $\ep < \sigma^2$ for all $x$. Since $g(x)$ is integrable, the dominated convergence theorem implies that $\lim_{\ep\to 0}\h(V_1)=\h(Y_1)$.

We then compute $\h(V_1|U_1)=\frac{1-\delta}{2} \log_2 (2 \pi e (g_1 \ep+1))+\frac{\delta}{2} \log_2 (2 \pi e (g_1 \sigma^2+1))$ implying that $\lim_{\ep \to 0} \h(V_1|U_1)= \frac{1-\delta}{2} \log_2 (2 \pi e )+\frac{\delta}{2} \log_2 (2 \pi e (g_1 \sigma^2+1))=\h(Y_1|B)$. As a result, (\ref{eq:lim1}) follows.


Finally, we note that the distribution of $V_D$ does not depend on $\ep$ so by direct calculation we obtain that $\I(U_2; V_D| U_1) = \I(X_S; Y_D| B)$ and  $\I(U_1, U_2; V_D) = \I(B, X_S; Y_D)$ for any $\ep$, and in particular, in the limits (\ref{eq:lim2}) and (\ref{eq:lim3}).
\hspace*{\fill}\QED


\section{Proof of \Tref{thm:TSPDF}}
\label{app:proofTSPDF}

In \Tref{thm:SPDF}, let $K=2$, $U_1 \sim \text{Uniform}[0,1]$. 
Set  $U_2|\{U_1 \in [0, 1-\delta_1-\delta_2] \} \sim \cN(0,\epsilon)$,  $U_2|\{U_1 \in (1-\delta_1-\delta_2,1-\delta_2] \} \sim \cN(0,\sigma_1^2)$, $U_2|\{ U_1 \in (1-\delta_2, 1] \} \sim \cN(0,\sigma_2^2)$,
 for $\delta_1, \delta_2 \ge 0$, $0< \delta_1+\delta_2 \le 1$, $0 < \epsilon < 1$ and $\delta_1 \sigma_1^2 + \delta_2 \sigma_2^2 + (1-\delta_1-\delta_2) \epsilon < 1$. Also let $f(1)=\ldots=f(N)=1$, and define
\[
w_i(U_1, V_i)=\left\{ 
\begin{array}{ll} 
0 & ; \ \text{if } U_1\in [0 , 1 - \delta_1 - \delta_2] \\ 
\sqrt{\kappa_{i1}} V_i & ; \ \text{if } U_1\in(1-\delta_1-\delta_2, 1-\delta_2] \\
\sqrt{\kappa_{i2}} V_i & ; \ \text{if } U_1\in(1-\delta_2, 1] 
\end{array}
\right.,
\]
with $\delta_1 \kappa_{i1}(g_i \sigma_1^2+1)+ \delta_2 \kappa_{i2}(g_i \sigma_2^2+1)< 1$, $i=1,\ldots,N$.
Then, $\tilde R= \tilde R_1 + \tilde R_2$ is achievable if
\begin{align}
\label{eq:cor_TSPDF_1}
\tilde R_1 & < \I(U_1; V_{1}), \\
\label{eq:cor_TSPDF_2}
\tilde R_2 & < \I(U_2; V_D|U_1), \\
\label{eq:cor_TSPDF_3}
\tilde R_1 + \tilde R_2 & < \I(U_1, U_2; V_{D}).
\end{align}
As was done in the proof of \Tref{thm:BSPDF}, by choosing an appropriate integrable upper bound for $|P_{V_1}(x) \log_2 (P_{V_1}(x))|$ that is independent of $\ep$, we can use the dominated convergence theorem and direct calculations to show that
\begin{align*}
\lim_{\ep \to 0} \I(U_1; V_1) =& \I(T;Y_1), \\
\lim_{\ep \to 0} \I(U_2; V_D| U_1) =& \I(X_S; Y_D| T),\\
\lim_{\ep \to 0} \I(U_1, U_2; V_D) =& \I(T, X_S; Y_D).
\end{align*}

Thus, for any rate $R=R_1+R_2$ that satisfies the conditions of \Tref{thm:TSPDF} we can find an $\ep>0$ small enough such that $R_1 \le \tilde R_1$ and $R_2 \le \tilde R_2$
for some $\tilde R_1 $ and $\tilde R_2$ satisfying \eq{eq:cor_TSPDF_1}, \eq{eq:cor_TSPDF_2}, \eq{eq:cor_TSPDF_3}, and as a result, by \Tref{thm:SPDF} the rate $R = R_1+R_2$ is achievable.
\hspace*{\fill}\QED


\section{Proof of \Tref{thm:SPDF}}
\label{app:SPDF}

The following definition and basic properties of typical sets can be found in \cite{Cover} and will be used in the proof of \Tref{thm:SPDF}.

\begin{defn}
For a finite collection of real random variables 
$(X_1, X_2, \ldots, X_k)$ with joint probability density function
$f_\bX(x_1, x_2, \ldots, x_k)$ let $\bS$ denote an ordered subset of the random variables and consider
$n$ independent identically distributed copies of $\bS$. The $\ep$-typical
set $\cA^{(n)}_\ep$ of the $n$-sequences $(\bx_1^n, \bx_2^n, \ldots, \bx_k^n)$
is defined as
\begin{align*}
 \cA^{(n)}_\ep(X_1, & X_2, \ldots, X_k) = \\
& \left\{
(\bx_1^n, \bx_2^n, \ldots, \bx_k^n) :
\left| -\frac{1}{n} \log f_\bX(\bs^n) - \h(\bS) \right| < \ep,
\forall \bS \subseteq \{X_1, X_2, \ldots, X_k \} \right\}.
\end{align*}
\end{defn}

\begin{prop}
\label{prop:A1}
For any $\ep >0 $, and for sufficiently large $n$, 
\begin{equation}
\Pr \left\{ \cA_{\ep}^{(n)} (\bS) \right\} \ge 1 - \ep, \ 
\forall \bS \subseteq \{X_1, X_2, \ldots, X_k \} .
\end{equation}
\end{prop}

\begin{prop}
\label{prop:A2}
Let $\cA_{\ep}^{(n)}(X_1,X_2)$ be the typical set corresponding to the probability
density function $f_{X_1,X_2}(x_1, x_2)$ and let $(\tilde{X}_1,\tilde{X}_2)$ have probability density function
$$
f_{\tilde{X}_1,\tilde{X}_2}(x_1, x_2) = f_{X_1}(x_1) f_{X_2}(x_2).
$$
Then
$$
\Pr\{ (\tilde{\bx}_1^n , \tilde{\bx}_2^n ) \in \cA_{\ep}^{(n)}(X_1,X_2) \} \doteq 2^{- n ( \I(X_1; X_2) \pm 3\ep)}.
$$
\end{prop}

\begin{prop}
\label{prop:A3}
For the probability
density function $f_{X_1,X_2,X_3}(x_1, x_2, x_3)$
let $\cA_{\ep}^{(n)}(X_1,X_2,X_3)$ be the typical set and let $(\tilde{X}_1,\tilde{X}_2, \tilde{X}_3)$ have probability density function
$$
f_{\tilde{X}_1,\tilde{X}_2,\tilde{X}_3}(x_1, x_2, x_3) = f_{X_1}(x_1) f_{X_2|X_1}(x_2|x_1) f_{X_3|X_1}(x_3|x_1).
$$
Then
$$
\Pr\{ (\tilde{\bx}_1^n , \tilde{\bx}_2^n,  \tilde{\bx}_3^n) \in \cA_{\ep}^{(n)}(X_1,X_2, X_3) \} \doteq 2^{- n ( \I(X_2; X_3|X_1) \pm 4\ep)}.
$$
\end{prop}

\begin{proof}(of \Tref{thm:SPDF})

We communicate over $b+1$ blocks of length $n$ each. In the first block, the source sends a message in $\{1,\ldots, \lfloor 2^{nR}\rfloor\}$ which is received by all relays. The relays remain silent during this first block. In block $m$, $1<m\le b$, the source encodes a new message in $\{1,\ldots, \lfloor 2^{nR}\rfloor\}$ which is received by all the relays, while the relays encode signals based on the signals received in block $m-1$. Finally, in block $b+1$ the source remains silent while the relays encode signals based on the signals received in block $b$. Assuming that the destination successfully decodes all the messages, the communication rate is $b/(b+1) R$ for large $n$, which tends to $R$ as $b\to \infty$. We next describe and analyze the coding scheme used in each block of length $n$. To simplify notation, we use index $t$ to represent the time index in block $m$ for the signal transmitted at the source and received at the relays, and the time index in block $m+1$ for the signal transmitted by the relays and received at the destination, for $m=1,\ldots, b$.

{\em Code construction:}
We generate a codebook $\cC_1$ with $2^{nR_1}$ random i.i.d. codewords chosen according to $P_{U_1}$. For each $\bu_{1,i_1}^n \in \cC_1$, we generate a codebook $\cC_{2,i_1}$ with $2^{nR_2}$ codewords with symbols $u_{2,i_1,i_2}[t]$ chosen independently according to $P_{U_2|U_1=u_{1,i_1}[t]}$. Proceeding similarly, for each $\bu_{1,i_1}^n \in \cC_1$, $\bu_{2,i_1,i_2}^n \in \cC_{2,i_1}$,..., $\bu_{k-1,i_1,\ldots,i_{k-1}}^n \in \cC_{k-1, i_1,i_2,\ldots,i_{k-2}}$ we generate a codebook $\cC_{k,i_1,\ldots,i_{k-1}}$ with $2^{nR_k}$ codewords with symbols $u_{k,i_1,\ldots,i_{k}}[t]$ chosen independently according to $P_{U_K|U_1=u_{1,i_1}[t],\ldots, U_{k-1}=u_{k-1,i_1,\ldots,i_{k-1}}[t]}$, for $k=3,\ldots,K$.

Each message $W  \in \{1,2,\ldots, M\}$ is represented by a $K$-tuple $(i_1,\ldots, i_K)$ where $i_k \in \{1,\ldots, \lfloor 2^{n R_k}\rfloor\}$. The source transmits message $(i_1,\ldots, i_K)$ by sending the codeword $\bu_{K,i_1,\ldots,i_K}^n \in \cC_{K,i_1,\ldots,i_{K-1}}$. That is,

$$
\bX_S^n=\text{Enc}_S(W) = \text{Enc}_S(i_1,\ldots,i_K)=(u_{K,i_1,\ldots,i_K}[1],\ldots,u_{K,i_1,\ldots,i_K}[n]).
$$

{\em Relay processing:}
At the end of each block $m\in\{1,\ldots, b\}$, relay $j$ attempts to decode the codewords $\bu_{1,i_1}^n$, $\bu_{2,i_1,i_2}^n$, ..., $\bu_{f(i),i_1,i_2,\ldots,i_{f(j)}}$ where $f(j)$ (as defined above) indicates the highest message level that the relay attempts to decode. Decoding consists of finding the unique set of codewords $\{\bu_{1,\hat{i}_1}^n, \bu_{2,\hat{i}_1,\hat{i}_2}^n, \ldots, \bu_{f(j),\hat{i}_1,\ldots,\hat{i}_{f(j)}}^n\}$ in codebooks $\cC_1, \cC_{2,\hat{i}_1},\ldots,\cC_{f(j),\hat{i}_1,\ldots,\hat{i}_{f(j)-1}}$ that are jointly typical with $\bY_j^n$, declaring an error if it finds zero or more than one such sets. 

Relay $j$ encodes the signal to be sent in transmission block $m+1$ based on the received signal $\bY_j^n$ and the decoded codewords $\{\bu_{k,i_1,\ldots,i_k}^n\}_{k=1}^{f(j)}$ in block $m$ as follows
\[
X_j[t]=\text{Enc}_{j}(\bY_j^n)[t]= w_{j}(u_{1,\hat{i}_1}[t],\ldots,u_{f(j),\hat{i}_1,\ldots,\hat{i}_{f(j)}}[t],Y_j[t]), \text{ for } t=1,\ldots, n.
\] 

{\em Decoding at the destination:}
The destination attempts to decode the codewords $\bu_{1,i_1}^n$, $\bu_{2,i_1,i_2}^n$, ..., $\bu_{K,i_1,i_2,\ldots,i_{K}}$ from the received signal in each block $\bY_D^n$ using joint typicality decoding as explained for the relays. Decoding consists of finding the unique set of codewords $\{\bu_{1,\hat{i}_1}^n, \bu_{2,\hat{i}_1,\hat{i}_2}^n, \ldots, \bu_{K,\hat{i}_1,\ldots,\hat{i}_K}^n\}$ in codebooks $\cC_1, \cC_{2,\hat{i}_1},\ldots,\cC_{K,\hat{i}_1,\ldots,\hat{i}_{K-1}}$ that are jointly typical with $\bY_D^n$, declaring an error if it finds zero or more than one such sets. The decoded message in each block is given by $\text{Dec}(\bY_D^n)=(\hat{i}_1,\ldots,\hat{i}_{K})$.

{\em Analysis of the probability of error:}
We find the average probability of error over all the random code constructions.
Due to the symmetry of the random code construction, the conditional 
probability of error is independent of the transmitted messages,
i.e. $P_e^{(n)} = \Pr\{ \text{Dec}(\bY_D^n) \neq (1,\ldots,1) \ |  \ \text{message }\ (1,\ldots,1) \ \text{was sent} \}$, and so WLOG we assume that $\bu_{K,1,1,\ldots,1}^n$
has been sent.

There is an error whenever some relay or the destination declares an error. 
For any given block, relay $j$ declares an error if the correct codewords $\{\bu_{1,1}^n,\ldots,\bu_{f(j),1,\ldots,1}^n\}$ are not jointly typical with $\bY_j^n$, or if incorrect codewords $\{\bu_{1,i_1}^n,\ldots,\bu_{f(j),i_1,\ldots,i_{f(j)}}^n\}$ are jointly typical with $\bY_j^n$ for some $(i_1,\ldots,i_{f(j)})\ne (1,\ldots, 1)$. Similarly, the destination declares an error if the correct codewords $\{\bu_{1,1}^n,\ldots,\bu_{K,1,\ldots,1}^n\}$ are not jointly typical with $\bY_D^n$, or if incorrect codewords $\{\bu_{1,i_1}^n,\ldots,\bu_{K,i_1,\ldots,i_{K}}^n\}$ are jointly typical with $\bY_D^n$ for some $(i_1,\ldots,i_{K})\ne (1,\ldots, 1)$.

In addition, to guarantee the power constraints at the source and the relays be satisfied, if the average power of the signal to be transmitted at the source is larger than one, or if the average power of the signal to be transmitted at some relay is larger than one, we declare an error and do not transmit the corresponding signals. 
 
Define the following events
\begin{align*}
E^R &= \{ \text{relay } j \text{ successfully decodes all messages up to level } f(j), \text{ for } j=1,\ldots, N \} \\
E^D_{i_1 i_2 \ldots i_K} &=  \left\{ (\bu_{1,i_1}^n, \bu_{2,i_1,i_2}^n, \ldots, \bu_{K,i_1,i_2,\ldots,i_K}^n, \bY_D^n) \in A^{(n)}_\ep(U_1,\ldots,U_K, V_{D}) 
| E^R \right\} \\
E^{R_j}_{i_1\ldots i_{f(j)}} &= \left\{ (\bu_{1,i_1}^n,\ldots, \bu_{f(j),i_1,\ldots,i_{f(j)}}^n, \bY_j^n ) \in A^{(n)}_\ep(U_1,\ldots,U_{f(j)}, V_{j})  \right\}, \text{ for } j=1,\ldots,N \\
E^{S}_p &=  \left\{ \frac{1}{n} \sum_{t=1}^n (u_{K,1,\ldots,1})[t])^2 \ge 1 \right\} \\
E^{R_j}_p & = \left\{ \frac{1}{n} \sum_{t=1}^n ( w_{j}(u_{1,1}[t],\ldots,u_{f(j),1,\ldots,1}[t],Y_j[t])^2 \ge 1 \right\} , \text{ for } j=1,\ldots,N.
\end{align*}
Then the probability of error can be upper bounded as follows using the union bound, 
\begin{align*}
P_e^{(n)}  = &\Pr(\text{Error} | E^R) \Pr(E^R) + \Pr(\text{Error} | (E^R)^c ) \Pr((E^R)^c) \\
\le & \Pr(\text{Error} | E^R) + \Pr((E^R)^c) \\
 = & 
\Pr \left[
(E^D_{1\ldots 1})^c \cup  
\cup_{(i_1 \ldots i_K) \neq(1 \ldots 1)} E^D_{i_1 \ldots i_K} \right] \\
& + \Pr\left[
E^{S}_p \cup \cup_{j=1}^N \left((E^{R_j}_{1 \ldots 1})^c \cup 
\cup_{(i_1 \ldots i_{f(j)})\neq (1 \ldots 1)} E^{R_j}_{i_1 \ldots i_{f(j)}} \cup 
E^{R_j}_p \right)\right] \\
 \le &
\Pr\bigg((E^D_{1 \ldots 1})^c\bigg) + \sum_{(i_1 \ldots i_K)\ne (1 \ldots 1)} \Pr\bigg(E^D_{i_1 \ldots i_K}\bigg) + 
\Pr\bigg(E^{S}_p\bigg)  \\
& + \sum_{j=1}^N \left[\Pr\bigg((E^{R_j}_{1\ldots 1})^c \bigg) + \sum_{(i_1 \ldots i_{f(j)}) \neq (1 \ldots 1)} \Pr\bigg(E^{R_j}_{i_1 \ldots i_{f(j)}}\bigg) + \Pr\bigg(E^{R_j}_p \bigg) \right] .
\end{align*}
By using \Pref{prop:A1} for large $n$ we have 
$\Pr((E^D_{1\ldots 1})^c) < \ep$, and $\Pr((E^{R_j}_{1 \ldots 1})^c ) < \ep $ for $j=1,\ldots,N$. 
By the law of large numbers, for $n$ large enough, we also have 
$\Pr(E^S_{p}) < \ep  $, and $ \Pr(E^{R_j}_{p}) < \ep$ for $j=1,\ldots,N$.

We bound the sum involving $\Pr\big(E^D_{i_1 \ldots i_K}\big)$ as follows,
\begin{align}
\sum_{(i_1 \ldots i_K)\ne (1 \ldots 1)} \Pr\big(E^D_{i_1 \ldots i_K}\big) =& \sum_{i_1\ne 1, i_2,\ldots, i_K} \Pr\big(E^D_{i_1 \ldots i_K}\big) + \sum_{i_2\ne 1, i_3,\ldots, i_K} \Pr\big(E^D_{1 i_2 \ldots i_K}\big)+\ldots+\sum_{i_K\ne 1} \Pr\big(E^D_{1 \ldots 1 i_K}\big) \nonumber\\
\le & 2^{n \sum_{k=1}^K R_k} 2^{-n (\I(U_1,\ldots, U_K; V_{D})-3\epsilon)}+2^{n \sum_{k=2}^K R_k} 2^{-n (\I(U_2,\ldots, U_K; V_{D}|U_1)-4\epsilon)}\nonumber\\
& + \ldots + 2^{n R_K} 2^{-n (\I(U_K; V_{D}|U_1,\ldots,U_{K-1})-4\epsilon)},
\label{eq:pred}
\end{align}
where we used the fact that $(\bu_{1,i_1}^n, \bu_{2,i_1,i_2}^n, \ldots, \bu_{K,i_1,i_2,\ldots,i_K}^n)$ is independent of $\bY_D^n$ when $i_1\ne 1$, the fact that $(\bu_{2,1,i_2}^n, \ldots, \bu_{K,1,i_2,\ldots,i_K}^n)$ is conditionally independent of $\bY_D^n$ conditioned on $\bu_{1,1}^n$ when $i_2\neq 1$, ..., the fact that $\bu_{K,1,\ldots,1,i_K}^n$ is conditionally independent of $\bY_D^n$ conditioned on $(\bu_{1,1}^n, \bu_{2,1,1}^n, \ldots, \bu_{K-1,1,\ldots,1}^n)$ when $i_K \neq 1$, and Propositions~\ref{prop:A2} and \ref{prop:A3}, valid for large $n$.

Similarly, we bound the sum involving $\Pr\big(E^{R_j}_{i_1 \ldots i_{f(j)}}\big)$ for $j=1,\ldots, N$, as follows,
\begin{align}
\sum_{(i_1 \ldots i_{f(j)})\ne (1 \ldots 1)} 
\Pr\big(E^{R_j}_{i_1 \ldots i_{f(j)}}\big) =& 
\sum_{i_1\ne 1, i_2,\ldots, i_{f(j)}} \Pr\big(E^{R_j}_{i_1 \ldots i_{f(j)}}\big) 
+ \sum_{i_2\ne 1, i_3,\ldots, i_{f(j)}} 
\Pr\big(E^{R_j}_{1 i_2 \ldots i_{f(j)}}\big)\nonumber\\
&+\ldots+\sum_{i_{f(j)}\ne 1} 
\Pr\big(E^{R_j}_{1 \ldots 1 i_{f(j)}}\big)\nonumber\\
 \le & 2^{n \sum_{k=1}^{f(j)} R_k} 2^{-n (\I(U_1,\ldots, U_{f(j)}; V_{j})-3\epsilon)}+2^{n \sum_{k=2}^{f(j)} R_k} 2^{-n (\I(U_2,\ldots, U_{f(j)}; V_{j}|U_1)-4\epsilon)}\nonumber\\
 & + \ldots + 2^{n R_{f(j)}} 2^{-n (\I(U_{f(j)}; V_{j}|U_1,\ldots,U_{f(j)-1})-4\epsilon)}.
 \label{eq:sumperror}
\end{align}

The exponent of each of the terms in \eq{eq:sumperror} gives a rate constraint that needs to be satisfied in order to have vanishingly small error probability as $n\to 0$. Due to the degradedness of the broadcast channel from the source to the relays, many of these rate constraints are redundant. The non-redundant rate constraints arising from \eq{eq:sumperror} for $j=1,\ldots, N$ together with the rate constraints arising from \eq{eq:pred} are given in (\ref{eq:thmSPDF}). 

Since $\ep$ is arbitrarily small, when the conditions (\ref{eq:thmSPDF}) are satisfied, there exists a sequence of codebooks with average probability of error going to zero
as $n$ tends to infinity. As a result, the rate $R=\sum_{k=1}^K R_k$ is achievable.
\end{proof}


\section{Proof of \Lref{lem:b_h}}
\label{apd:b_h}
We find the Taylor series of $\h(Y)$ for $\bar \delta$ around zero. 
We will use Leibniz integral rule for improper integrals~\cite{Leibniz}.
\begin{prop}[Leibniz integral rule]
\label{prop:leibniz}
Let $f(x,y)$ and its partial derivative $\frac{\partial}{\partial y} f(x,y)$ 
be continuous everywhere. If there exist functions $g_1(x)$ and $g_2(x)$ such that $|f(x,y)| \le g_1(x)$ ,  $|\frac{\partial}{\partial y} f(x,y)| \le g_2(x)$ for all $x$, $\int_{-\infty}^\infty g_1(x) dx < \infty$ and $\int_{-\infty}^\infty g_2(x) dx < \infty$, then
$$
\frac{d}{dy} \int_{-\infty}^{\infty} f(x,y) dx = 
\int_{-\infty}^{\infty} \frac{\partial f}{\partial y}(x,y) dx.
$$
\end{prop}
In order to apply \Pref{prop:leibniz} we need to verify that the conditions of the proposition apply to the function $f(y,\bar\delta)= f_Y(y,\bar\delta) \ln f_Y(y,\bar\delta)$, where 
$$
f_{Y}(y,\bar\delta) = \left(1 - \sum_{i=1}^{q-1} \delta_i \right) \, e^{-y^2/(2\sigma_0^2)}/(\sqrt{2\pi}\sigma_0)
+ \sum_{i=1}^{q-1} \delta_i \, e^{-y^2/(2\sigma_i^2)}/(\sqrt{2\pi}\sigma_i).
$$
It is easy to verify that $f(y,\bar\delta)$ and $\frac{\partial}{\partial \delta_i} f(y,\bar\delta)$ are continuous for $i=1,\ldots, q{-}1$. We next find a function $g(y)$ such that $|f(y,\bar\delta)|<g(y)$ and $\int_{-\infty}^\infty g(y) dy < \infty$. To that end we first find lower and upper bounds for $f_Y(y,\bar\delta)$ that are independent of $\bar\delta$. We assume that $\sum_{i=1}^{q-1} \delta_i < \delta_0 < 1$, and define $\sigma_{\min} = \min_{i\in\{0,\ldots,q-1\}} \sigma_i$ and $\sigma_{\max} = \max_{i\in\{0,\ldots,q-1\}} \sigma_i$. 
\[
\frac{1-\delta_0}{\sqrt{2\pi}\sigma_0} e^{-y^2/(2 \sigma_0^2)} \le f_{Y}(y,\bar\delta) \le \frac{q}{\sqrt{2\pi}\sigma_{\min}} e^{-y^2/(2 \sigma_{\max}^2)}
\]
We use these bounds to obtain
\[
g(y)=\frac{q}{\sqrt{2\pi}\sigma_{\min}} e^{-y^2/(2 \sigma_{\max}^2)} \left[\max\left\{\left|\ln\left(\frac{1-\delta_0}{\sqrt{2\pi} \sigma_0}\right)\right|, \left|\ln\left(\frac{q}{\sqrt{2 \pi}\sigma_{\min}}\right) \right| \right\}+\frac{y^2}{2\sigma_{\min}^2} \right],
\]
which is integrable.

For each $i=1,\ldots,q{-}1$, we find a function $g_i(y)$ such that $|\frac{\partial}{\partial \delta_i}f(y,\bar\delta)|<g_i(y)$ and $\int_{-\infty}^\infty g_i(y) dy < \infty$. Noting that
\[
\frac{\partial}{\partial \delta_i}f(y,\bar\delta) = \left(-\frac{1}{\sqrt{2\pi}\sigma_0} e^{-y^2/(2 \sigma_0^2)}+ \frac{1}{\sqrt{2\pi}\sigma_i} e^{-y^2/(2 \sigma_i^2)}\right)\left(\ln f_Y(y,\bar\delta) +1 \right)
\]
we obtain
\[
g_i(y)=\left(\frac{1}{\sqrt{2\pi}\sigma_0} e^{-y^2/(2 \sigma_0^2)}+ \frac{1}{\sqrt{2\pi}\sigma_i} e^{-y^2/(2 \sigma_i^2)}\right) 
\left[\max\left\{\left|\ln\left(\frac{1-\delta_0}{\sqrt{2\pi} \sigma_0}\right)\right|, \left|\ln\left(\frac{q}{\sqrt{2 \pi}\sigma_{\min}}\right) \right| \right\}+\frac{y^2}{2\sigma_{\min}^2} + 1\right],
\]
which is also integrable.

Having verified the conditions of \Pref{prop:leibniz}, we can exchange the order of derivation and integration to compute the Taylor series of $\h(Y)$ as follows,
\begin{align*}
\h(Y) & = \h(Y)|_{\bar \delta = 0} +
\sum_{i=1}^{q-1} \delta_i \, \Big( \frac{\partial \h(Y )}{\partial \delta_i}
\Big)|_{ {\bar{\delta}} = 0} + O(\norm{ {\bar{\delta}} }^2_2) 
\\
 & = 
 \Big( \int_{-\infty}^{\infty} -f_{Y}(y) \log_2[ f_{Y}(y) ] \, dy 
 \Big) \Big|_{ {\bar{\delta}} =0} +
 \sum_{i=1}^{q-1} \delta_i \Big( \frac{\partial}{\partial \delta_i} 
 \int_{-\infty}^\infty -f_{Y}(y) \log_2[f_{Y}(y)] \, dy \Big) \Big|_{ {\bar{\delta}} = 0} +O(\norm{ {\bar \delta} }^2_2)
 \\
& = 
\int_{-\infty}^{\infty}
\Big(  - f_{Y}(y) \log_2 [ f_{Y}(y) ] \Big) \Big|_{\bar \delta = 0} \, dy +
\sum_{i=1}^{q-1} \delta_i
\int_{-\infty}^{\infty}
\Big(
- \frac{\partial}{\partial \delta_i} f_{Y}(y) \log_2 [  f_{Y}(y) ]
\Big) \Big|_{\bar \delta = 0} \, dy + O(\norm{ \bar \delta }^2_2) 
\\
& = 
\frac{1}{2}\log_2[2\pi e \sigma_0^2] + 
\sum_{i=1}^{q-1} \delta_i( \sigma_i^2/\sigma_0^2 - 1 )/( 2\ln 2 )
+ O(\norm{ \bar \delta }^2_2).
\end{align*}
\hspace*{\fill}\QED


\section{Proof of \Tref{thm:asymptTSPDF}}
\label{app:asympTSPDFproof}

In \Tref{thm:TSPDF} we set $\delta_1 = \beta_1 \sqrt{g h}$,
$\delta_2 = \beta_2 \sqrt{g h}$, $\kappa_1\dfn h\kappa_{11}= h\kappa_{21}$,
$\kappa_2 \dfn h\kappa_{12} =h\kappa_{22}$ and we set $\gamma_1 \dfn h \sigma_1^2$, 
$\gamma_2 \dfn h \sigma_2^2$, for positive constants $\beta_1, \beta_2, 
\gamma_2, \gamma_2, \kappa_1, \kappa_2$ in $\R^+$. We know the rate $R_1+R_2$ is achievable if
\begin{align*}
R_1 & < \I(T; Y_1), \\
R_2 & < \I(X_S; Y_D | T), \\
R_1+R_2 & < \I(T, X_S; Y_D) = \I(T; Y_D) + \I(X_S; Y_D|T).
\end{align*}
We also can achieve the maximum sum rate $R_1+R_2$ if
\begin{align*}
R_1 & < \I(T; Y_1), \\
R_1 & < \I(T; Y_D), \\
R_2 & < \I(X_S; Y_D|T).
\end{align*}
We compute the mutual information expressions 
$\I(T; Y_1)$ and $\I(X_S; Y_D | T)$ for small $h$. 

The probability distribution of $Y_1$ conditioned on 
$T =1$ and $T =2$ is equal
to $\cN(0, g \sigma_1^2 + 1)$ and $\cN(0, g \sigma_2^2 + 1)$ respectively.
For the case that $T = 0$ the distribution of $Y_1$ is $\cN(0,1)$.
By applying \Cref{cor:I} we have
\begin{equation}
\I(T;Y_1) = 
\frac{h \sqrt{g/h}}{2\ln2} \left(
\beta_1 \left( \gamma_1 g/h  - \ln [1+ \gamma_1 g/h ] \right) +
\beta_2 \left( \gamma_2 g/h  - \ln [1+ \gamma_2 g/h ] \right)
\right) + O(h^2),
\end{equation}
and 
\begin{align*}
\I(X_S ; Y_D | T ) & = 
\frac{\delta_1}{2} \log_2\left( 
1 +  \frac{4 g \sigma_1^2 \kappa_1 }{2 \kappa_1 + 1} \right) +
\frac{\delta_2}{2} \log_2\left( 
1 +  \frac{4 g \sigma_2^2 \kappa_2 }{2 \kappa_2 + 1} \right) \\
& = 
\frac{h \left(\sqrt{g/h } \ln \left[1+\frac{4 (g/h)  \gamma _1 \kappa _1}{1+2 \kappa _1}\right] \beta _1+\sqrt{g/h} \ln \left[1+\frac{4 (g/h)  \gamma _2 \kappa _2}{1+2 \kappa _2}\right] \beta _2\right)}{2 \ln2} 
+ O(h^2).
\end{align*}
Finally to compute $\I(T;Y_D)$, we note that the 
distribution of $Y_D$ conditioned on $T$ being $0,1,$ and $2$ 
is equal to $\cN(0,1)$, $\cN(0,4 g \sigma^2_1 \kappa_1 + 2 \kappa_1 + 1)$ and $\cN(0, 4 g \sigma^2_2 \kappa_2 + 2 \kappa_2 + 1)$ respectively. 
Then, by \Cref{cor:I} we have
\begin{align}
\nonumber
\I(T; Y_D)  = &
\frac{h \sqrt{g/h}}{2 \ln 2} \Big(
\beta_1 (4 (g/h) \gamma_1 \kappa_1 + 2 \kappa_1  - \ln[1 + 4 (g/h) \gamma_1 \kappa_1 + 2 \kappa_1] ) + \\
& \quad
\beta_2 (4 (g/h) \gamma_2 \kappa_2 + 2 \kappa_2  - \ln[1 + 4 (g/h) \gamma_2 \kappa_2 + 2 \kappa_2] ) \Big) + O(h^2).
\end{align}
To complete the proof we need to impose the power constraints at the source
and relays:
\begin{align*}
\delta_1 \sigma_1^2 + \delta_2 \sigma_2^2 &=  
\sqrt{g/h} (\gamma_1 \beta_1 + \gamma_2 \beta_2) < 1, \nonumber\\
\delta_1 \kappa_{i1}(g \sigma_1^2+1) + \delta_2 \kappa_{i2}(g \sigma_2^2+1) &=
\sqrt{g/h}(\beta_1(1+ \gamma_1 g/h )\kappa_1 + 
\beta_2(1+ \gamma_2 g/h)\kappa_2 ) < 1. \nonumber
\end{align*}
\hspace*{\fill}\QED


\section{Proof of \Tref{thm:N_relay}}
\label{apd:N_relay}


The derivation of the achievable rates in the low SNR regime for BSPDF coding in the symmetric parallel relay network with $N$ relays is similar to the case of the two relay network. 
The difference is in the distribution of the signal $Y_{D}$ 
which is the combination of signals from $N$ relays instead of two relays.
By applying \Tref{thm:BSPDF} we are able to achieve any rate $R_1 + R_2$ such that
\begin{align}
\label{eq:ap_in_1}
R_1 & < \I(B; Y_1), \\
\label{eq:ap_in_2}
R_2 & < \I(X_S; Y_{D} |B) \\
\label{eq:ap_in_3}
R_1 + R_2 & < \I(B, X_S; Y_{D}) = \I(B; Y_D) + \I(X_S; Y_D|B)
\end{align}
Therefore, the maximum sum rate $R_1+R_2$ that satisfies the conditions \eq{eq:ap_in_1}, \eq{eq:ap_in_2}, and \eq{eq:ap_in_3} also satisfies the conditions
\begin{align*}
R_1 & < \I(B; Y_1), \\
R_1 & < \I(B; Y_D), \\
R_2 & < \I(X_S; Y_{D} | B).
\end{align*}

In \Tref{thm:BSPDF}, 
we set $\delta = \beta \sqrt{g h}$ such that
$0 < \delta \le 1$ and $\kappa = \kappa_i < 1/(g + \delta)$
for $i = 1,2,\ldots, N$. 
Let $\sigma^2 = \sqrt{g/h}/\beta + 1$. Conditioned
on $B = 1$, the random variable $Y_1$ has
distribution $\cN(0,\sigma^2)$ and
conditioned on $B = 0$ we have $Y_1 \sim \cN(0,1)$. Thus
by \Cref{cor:I}
$$
\I(B; Y_1) = h \left( g/h + \sqrt{g/h}\beta
\ln \left( \frac{\beta}{\sqrt{g/h} + \beta} \right) \right) / (2\ln 2)
+ O(h^2).
$$

At the destination, optimizing over the value of
$\kappa$ we get $\kappa = 1/(g+\delta)$. If we define
$\sigma_D^2 = N^2 \sqrt{g/h}/(\beta (g/h + \beta \sqrt{g/h}))
+ N / (g/h + \beta \sqrt{g/h}) + 1$, then conditioned on
$B = 1$ we have $Y_{D} \sim \cN(0, \sigma_D^2)$ and
conditioned on $B = 0 $, we have $Y_{D} \sim \cN(0,1)$, thus using \Cref{cor:I}
\begin{align*}
\I(B, Y_{D}) & = 
h \beta \sqrt{g/h} \Bigg(
N^2 \sqrt{g/h}/(\beta (g/h + \beta \sqrt{g/h}))
+ N / (g/h + \beta \sqrt{g/h}) 
\\
& - \ln \Big[
1 + N^2 \sqrt{g/h}/(\beta (g/h + \beta \sqrt{g/h}))
+ N / (g/h + \beta \sqrt{g/h}) 
\Big]
\Bigg)/(2\ln 2) + O(h^2)
\end{align*}
and 
\begin{align*}
\I(X_S; Y_{D} | B) 
& = h \beta \sqrt{g/h} \log_2 \left[
1 + \frac{N^2 \sqrt{g/h} /(\beta (g/h + \beta \sqrt{g/h})}
{1 + N/(g/h + \beta\sqrt{g/h})} 
\right].
\end{align*}
After some simplification the theorem follows.
\hspace*{\fill}\QED


\section{Proof of \Tref{thm:asym_BSPDF}}
\label{apd:asym_BSPDF}


Using \Tref{thm:BSPDF}, we can achieve the rate $R_1+R_2$ if
\begin{align*}
R_1 & < \I(B; Y_1) \\
R_2 & < \I(X_S; Y_D | B) \\
R_1 + R_2 & < \I(B, X_S; Y_D) = \I(B; Y_D) + \I(X_S; Y_D | B).
\end{align*}
Equivalently, we can achieve the maximum sum rate $R_1+R_2$ if
\begin{align*}
R_1 & < \I(B; Y_1) \\
R_1 & < \I(B; Y_D) \\
R_2 & < \I(X_S; Y_D | B).
\end{align*}
We set $\delta = \beta\sqrt{g h}$ for
a positive constant $\beta$ such that $0 < \beta \sqrt{gh} \le 1$, 
$k_1 \dfn h\kappa_1 < h/( g+\delta )$ and 
$k_2 \dfn h\kappa_2 < h/( h + \delta )$. 
Let $\sigma^2_1 = \sqrt{g/h}/\beta + 1$ and
$\sigma^2_2 = 1/(\beta \sqrt{g/h}) + 1$.
Conditioned on $B=1$ the random variable 
$Y_1$ has distribution $\cN(0,\sigma_1^2)$. Conditioned on
$B = 0$ we have $Y_1 \sim \cN(0,1 )$. Then by \Cref{cor:I} we
have
\begin{align*}
\I(B; Y_1) & = 
h \beta \sqrt{g/h}\big(\sqrt{g/h}/\beta - \ln(\sqrt{g/h}/\beta + 1) \big) / (2\ln 2)
+ O(h^2).
\end{align*}

For the destination defining
$\sigma_3^2 = 1\sqrt{g/h} (\sqrt{k_1} + \sqrt{k_2})^2/\beta + k_1 + (g/h) k_2 + 1$, conditioned on $B = 1$
we have $Y_D \sim \cN(0, \sigma_3^2)$ and conditioned on
$B = 0$ we have $Y_D \sim \cN(0, 1)$. Therefore
\begin{align*}
\I(B & ;Y_D)  = \\
& \frac{h \sqrt{g / h} \beta}{2\ln 2} \Bigg( 
\sqrt{\frac{g}{h}} (\sqrt{k_1} + \sqrt{k_2})^2/\beta + k_1 +\frac{g}{h} k_2 - \ln \Big[1 +
\sqrt{\frac{g}{h}} (\sqrt{k_1} + \sqrt{k_2})^2/\beta + k_1 + \frac{g}{h} k_2 
\Big]
\Bigg)
+ O(h^2).
\end{align*}
Finally by \Cref{cor:I}, 
\begin{align*}
\I(X_S; Y_D | B)  
= & \frac{h \beta \sqrt{g/h}}{2} \log_2
\left(
1 + \frac{\sqrt{g/h} (\sqrt{k_1} + \sqrt{k_2})^2/\beta}
{1 + k_1 + (g/h) k_2} 
\right),
\end{align*}
which completes the proof. 
\hspace*{\fill}\QED


\section{Proof of \Tref{thm:BSPDF+DF}}
\label{apd:BSPDF+DF}
 

We use two layers of coding requiring relay~$2$ and the destination to decode both layers of the code and while requiring relay~$1$ to decode only the first
layer. By \Cref{cor:BSPDF2} we can show that the rate
$R_1 + R_2$ is achievable if 
\begin{align*}
R_1 & < \I(B; Y_1), \\
R_2 & < \I(X_S; Y_2 | B), \\
R_2 & < \I(X_S; Y_D | B), \\
R_1 + R_2 & < \I(B, X_S; Y_D) = \I(B; Y_D) + \I(X_S; Y_D | B).
\end{align*}
Equivalently, the following conditions are sufficient to achieve the maximum sum rate $R_1 + R_2$,
\begin{align*}
R_1 & < \I(B; Y_1), \\
R_1 & < \I(B; Y_D), \\
R_2 & < \I(X_S; Y_2 | B), \\ 
R_2 & < \I(X_S; Y_D | B).
\end{align*}
In the following we compute, $\I(B; Y_1)$, 
$\I(B, Y_D)$,
$\I(X_S;Y_2|B)$, and $\I(X_S; Y_D | B)$. 
Set $\delta = \beta \sqrt{g h}$, and define $\kappa = \kappa_1 h$, assuming that $ \delta (\kappa / h) (1 + g /\delta) < 1$ in order to satisfy the power constraint at relay 1. 

Conditioned on $B = 1$ we have
$Y_1 \sim \cN(0, \sigma_1^2)$, $Y_2 \sim \cN(0, \sigma_2^2)$, and
$Y_D \sim \cN(0, \sigma_D^2)$, where $\sigma_1^2 = \sqrt{g/h}/\beta + 1$,
$\sigma_2^2 = 1/(\beta \sqrt{g/h}) + 1$, and
$\sigma_D^2 = \sqrt{g/h}(1+\sqrt{\kappa})^2/\beta + \kappa + 1$. Conditioned on $B = 0$
we have $Y_1 \sim \cN(0, 1)$, $Y_2 \sim \cN(0, 1)$, and $Y_D \sim \cN(0, 1)$. 

By applying \Cref{cor:I} we get
\begin{align*}
\I(B; Y_1) & =
h \beta \sqrt{g/h} \Big(
\sqrt{g/h}/\beta  - \ln \Big[
1 + \sqrt{g/h}/\beta  \Big] \Big) / (2\ln 2) + O(h^2), \\
\I(B; Y_D) & = 
h \beta \sqrt{g/h} \Big(
\sqrt{g/h}(1+\sqrt{\kappa})^2/\beta + \kappa - \ln \Big[
1 + \sqrt{g/h}(1+\sqrt{\kappa})^2/\beta + \kappa \Big] \Big) / (2\ln 2)
+ O(h^2).
\end{align*}
To complete the proof we need to compute $\I(X_S; Y_2 | B)$ and
$\I(X_S; Y_D | B)$. We have by \Cref{cor:I}
\begin{align*}
\I(X_S; Y_2 | B) 
& = \frac{h \beta \sqrt{g/h}}{2} \log_2 \left(
1 + \frac{1}{\beta \sqrt{g/h}}  \right), \\
\I(X_S; Y_D | B) 
& = \frac{h \beta \sqrt{g/h}}{2} \log_2 \left(
1 + \frac{\sqrt{g/h} (1 + \sqrt{\kappa})^2 / \beta}{\kappa + 1} 
\right).
\end{align*}
\hspace*{\fill}\QED


\section{Derivation of the lower bound on the minimum energy-per-bit}
\label{apd:ebno_upper_bound}

In order to find a lower bound on the minimum energy-per-bit we normalize
the channel gains as $\tilde g = P g /N_0$ and $\tilde h = \gamma P h / N_0$ 
\begin{align}
\cE_b & \ge \inf_{\gamma \ge 0} \lim_{P \rightarrow 0}
\frac{(2\gamma+1)P}{\Cupper} \\
& \ge \inf_{\gamma > 0}
\frac{(1+ 2\gamma) N_0 \ln 2 }
{\max_{0 \le \rho \le 1} \min
\{ g, (g+ \gamma h(1-\rho^2))/2, \gamma h(1+\rho) \} } \\
& \ge
\frac{N_0 \ln 2 }
{\sup_{\gamma>0} \max_{0 \le \rho \le 1} \min
\left\{ 
\frac{g}{2\gamma+1}, 
\frac{g  + \gamma h(1-\rho^2)}{2(2\gamma+1)}, 
\frac{\gamma h(1+\rho)}{2\gamma+1} \right\} }.
\label{eq:ebno_upper_1}
\end{align}
For a fixed $\gamma$ we can solve the maximization over $\rho$
\begin{align*}
\max_{0 \le \rho \le 1} & \min\{ g, (g+ \gamma h(1-\rho^2))/2, \gamma h(1+\rho) \} =
\begin{cases}
2 \gamma h & ; \  0 < \gamma h /g < 1/4 \\
\sqrt{\gamma g h} & ; \  1/4 \le \gamma h /g \le 1 \\
g & ;  \ 1 < \gamma h /g 
\end{cases},
\end{align*}
and therefore, 
\begin{align*}
\sup_{\gamma>0} \max_{0 \le \rho \le 1} \min
\left\{ 
\frac{g}{2\gamma+1}, 
\frac{g  + \gamma h(1-\rho^2)}{2(2\gamma+1)}, 
\frac{\gamma h(1+\rho)}{2\gamma+1} \right\} =& \max \left\{
\max_{0 < \gamma < \frac{g}{4h}} \frac{2\gamma h}{2\gamma+1},
\max_{\frac{g}{4h} \le \gamma \le \frac{g}{h} } \frac{\sqrt{\gamma g h}}{2\gamma+1},
\max_{\frac{g}{h} < \gamma} \frac{g}{2\gamma+1} \right\} \\
 = & \max \left\{
\frac{gh}{g +2h},\max_{\frac{g}{4h} \le \gamma \le \frac{g}{h}} \frac{\sqrt{\gamma g h}}{2\gamma+1},\frac{gh}{2g + h} \right\} \\
= &
\begin{cases}
\max \left\{
\frac{gh}{g +2h}, \sqrt{\frac{gh}{8}},\frac{gh}{2g + h} \right\}  & ; \ 1/2 \le g/h  \le 2 \\
\max \left\{
\frac{gh}{g +2h}, \frac{gh}{2g + h}\right\} & ; \ \text{otherwise}
\end{cases} \\
= &
\begin{cases}
\frac{gh}{2 g +h} & ; \ 0 < g/h < 1/2 \\
\sqrt{\frac{gh}{8}} & ; \ 1/2 \le g/h \le 2 \\
\frac{gh}{g + 2h} & ; \ 2 < g/h
\end{cases}.
\end{align*}


\end{document}